\documentclass[11pt]{article}
\usepackage[margin=1in]{geometry}
\usepackage[utf8]{inputenc}
\usepackage[english]{babel}
\usepackage[cmex10]{amsmath}
\usepackage{todonotes}
\usepackage{amsthm}
\usepackage{amssymb}
\usepackage{xspace}
\usepackage{hyperref}
\usepackage{cite}
\usepackage{comment}
\usepackage{boxedminipage}
\newcommand{\cO}{{\mathcal{O}}}
\newcommand{\term}[1]{#1}
\newcommand{\yes}{{yes}}

\newcommand{\yesinstance}{\yes-instance\xspace}

\newcommand{\yesinstances}{\yes-instances\xspace}

\newcommand{\twetamod}{treewidth-$\eta$-modulator\xspace}

\usepackage{titling}
\newcommand*\samethanks[1][\value{footnote}]{\footnotemark[#1]}

\thanksmarkseries{arabic}

\newcommand{\constPower}{\mu}

\newcommand{\consteps}{\varepsilon}
\newcommand{\constbidimens}{\beta}

\newtheorem{theorem}{Theorem}[section]
\newtheorem{lemma}[theorem]{Lemma}

\newtheorem{proposition}[theorem]{Proposition}
\newtheorem{observation}[theorem]{Observation}
\theoremstyle{definition}
\newtheorem{definition}[theorem]{Definition}

\newcommand{\tw}{\operatorname{tw}}

\usepackage{color}

\definecolor{Black}{rgb}{0,0, 0}
\definecolor{mygreen}{rgb}{0, .6, 0}
\definecolor{Blue}{rgb}{0, 0 ,1}
\definecolor{Red}{rgb}{1, 0 ,0}
\definecolor{Other}{rgb}{.1, .6,.7}
\definecolor{Otherother}{rgb}{.9, .3,.4}
\definecolor{Brown}{rgb}{0.5, 0.3, 0.3}
\definecolor{Magenta}{rgb}{0.5, 0, 1}
\definecolor{White}{rgb}{1, 1, 1}

\newcommand{\ProblemFormat}[1]{{\sc #1}}

\newcommand{\ProblemIndex}[1]{\index{problem!\ProblemFormat{#1}}}

\newcommand{\ProblemName}[1]{\ProblemFormat{#1}\ProblemIndex{#1}\xspace}

\newcommand{\probVC}{\ProblemName{Vertex Cover}}

\newcommand{\probFVS}{\ProblemName{Feedback Vertex Set}}

\newcommand{\probCycPacking}{\ProblemName{Cycle Packing}}

\newcommand{\probIS}{\ProblemName{Independent Set}}

\newcommand{\probrDS}{\ProblemName{$r$-Dominating Set}}
\newcommand{\probrSC}{\ProblemName{$r$-Scattered Set}}

\newcommand{\probCDS}{\ProblemName{Connected Dominating Set}}

\newcommand{\probCVC}{\ProblemName{Connected Vertex Cover}}

\begin{document}

\title{Bidimensionality and Kernels\thanks{Part of the results of this paper have appeared in \cite{F.V.Fomin:2010oq}.}~$^,$\thanks{Emails of authors: \texttt{fomin@ii.uib.no}, \texttt{daniello@ucsb.edu}, \texttt{saket@imsc.res.in}/\texttt{Saket.Saurabh@ii.uib.no}, \texttt{sedthilk@thilikos.info}.}}
\author{
  Fedor V. Fomin\thanks{Department of Informatics, University of Bergen, Norway.}~$^,$\thanks{Supported by the European Research Council under the European Union’s Seventh Framework Programme (FP/2007-2013) / ERC Grant Agreements no. 267959 and by the Norwegian Research Council via MULTIVAL project.}~$^,$\thanks{Supported by the Research Council of Norway and the French Ministry of Europe and Foreign Affairs, via the Franco-Norwegian project PHC AURORA 2019.} 
 \and
  Daniel Lokshtanov\thanks{University of California Santa Barbara, Goleta, CA, USA. }
  \and 
  Saket Saurabh\samethanks[3]~$^,$\thanks{The Institute of Mathematical Sciences, HBNI, Chennai, India.}~$^,$\thanks{Supported by the European Research Council under the European Union’s Seventh Framework Programme (FP/2007-2013) / ERC Grant Agreements no. 306992.}
    \and 
    Dimitrios M. Thilikos\samethanks[5]~$^,$\thanks{LIRMM, Univ Montpellier, CNRS, Montpellier, France.}~$^,$\thanks{Supported by projects DEMOGRAPH (ANR-16-CE40-0028) and ESIGMA (ANR-17-CE23-0010).}}

\date{}

\maketitle

\begin{abstract}

\noindent  {\sl Bidimensionality Theory}   was introduced by [{\sc E.~D. Demaine, F.~V. Fomin, M.~Hajiaghayi, and D.~M. Thilikos}. {\em  Subexponential
  parameterized algorithms on graphs of bounded genus and {$H$}-minor-free
  graphs}, J. ACM, 52 (2005), pp.~866--893] as a tool to obtain {\sl sub-exponential} time parameterized algorithms  on $H$-minor-free graphs. In [{\sc E.~D. Demaine and M.~Hajiaghayi}, {\em Bidimensionality: new connections
  between {FPT} algorithms and {PTAS}s}, in Proceedings of the 16th Annual
  ACM-SIAM Symposium on Discrete Algorithms (SODA), SIAM, 2005, pp.~590--601]  this theory was extended in order to obtain polynomial time approximation schemes (PTASs) for bidimensional problems. In this work, we establish a third meta-algorithmic direction for bidimensionality theory by relating it to the existence of {  linear kernels} for parameterized problems.
 In particular, we prove that every minor (respectively contraction) bidimensional problem that satisfies a separation property and is 
 expressible in Countable Monadic Second Order Logic (CMSO), admits a linear kernel for classes of graphs that exclude a fixed graph (respectively an apex graph) $H$ as a minor. 
%
%
%
 Our results imply that a multitude of bidimensional problems
  admit linear kernels on the corresponding graph classes. For most of these problems no polynomial kernels on $H$-minor-free graphs were known prior to our work.

\end{abstract}

\noindent{\bf Keywords:} Kernelization, Parameterized algorithms,  Treewidth, Bidimensionality

\section{Introduction}\label{sec_intro}

  Bidimensionality theory was  introduced by Demaine et al. in~\cite{DemaineFHT05jacm}.
  This theory is build on  cornerstone theorems from Graph Minors Theory of Robertson and Seymour \cite{RobertsonS-V} 
   and  initially  it was developed to  unify and extend subexponential fixed-parameter algorithms
  for {\sf NP}-hard graph problems to a  broad range of graphs including planar graphs,  map graphs, bounded-genus graphs and graphs excluding any fixed graph as a minor     \cite{DemaineFHT05jacm,DemaineH07-CJ,DemaineFHT05sidma, DemaineFHT05talg,DemaineH05II} (see also~\cite{FominLRS10,FominLS18excl,GrigorievKT14bidi,BasteT17cont} for other graph classes).
 Roughly speaking, the problem is bidimensional if the solution value for  the problem on a $k\times k$-grid is $\Omega(k^2)$, and  contraction/removal of edges does not increase solution value.  Many natural problems are bidimensional, including
{\sc Dominating Set}, {\sc Feedback Vertex Set}, {\sc Edge Dominating Set}, {\sc Vertex
Cover}, {\sc $r$-Dominating Set}, {\sc Connected Dominating Set}, {\sc Cycle Packing},
{\sc Connected Vertex
Cover}, {\sc  Graph Metric TSP}, and many others.  

 The second application of bidimensionality was given by Demaine and Hajiaghayi  
 in \cite{DemaineHaj05}, where it has been shown that bidimensionality is a useful theory not only in the design of fast fixed-parameter algorithms  but also in the  design of fast PTASs.   Demaine and Hajiaghayi established a  link between parameterized and approximation algorithms by   proving that every bidimensional problem satisfying some simple separation properties has a PTAS on planar graphs and other classes of sparse graphs. See also \cite{FominLRS10,FominLS18excl} for further development of the applications of bidimensionality in the theory of EPTASs.
We refer to the surveys~\cite{DemaineH07-CJ,DornFT08-csr,Thilikos15bidi} for further information on bidimensionality and its applications, as well as the  book \cite{Cygan15_book}.  

In this work we give the third application of bidimensionality, that is {\sl kernelization}.
Kernelization  can be seen as
 the strategy of analyzing  preprocessing or data reduction heuristics from a parameterized complexity perspective.
  Parameterized complexity introduced by Downey and Fellows is basically a two-dimensional generalization of ``{\sf P}\!~vs.\!~{\sf NP}''  where, in addition
to the overall input size $n$, one studies the effects on computational complexity of a secondary measurement
that captures additional relevant information.  This additional information can be the solution size or the quantification of some structural restriction on the input, such as   the treewidth or the genus 
of the input graph.  The secondary information is quantified by a positive integer $k$ and is
called the {\em parameter}.  Parameterization can be deployed in many different ways; for general background
on the theory see \cite{Cygan15_book,DowneyFbook13,FlumGrohebook,Niedermeierbook06}.

A parameterized  problem with a parameter $k$ is said to admit a {\it polynomial kernel} if there is a polynomial time algorithm (the degree of polynomial is independent of $k$), called a {\em kernelization} algorithm, that reduces the input instance down to an instance with size bounded by a polynomial $p(k)$ in $k$, while preserving the answer.
Kernelization has been extensively studied in parameterized complexity,  resulting in polynomial kernels for a variety of problems. Notable examples of known kernels are a $2k$-sized vertex kernel for {\sc Vertex Cover}~\cite{ChenKJ01}, a $355k$ vertex-kernel for {\sc Dominating Set} on planar graphs~\cite{AFN04}, which later was improved to a $67k$ vertex-kernel~\cite{ChenFKX07}, or an $\cO(k^2)$ kernel for {\sc Feedback Vertex Set}~\cite{fvs-kernel:talg}  parameterized by the solution size.
One of the  most intensively studied  directions in kernelization is
the study of problems on planar graphs and other classes of sparse graphs. This study was initiated by
  Alber et al.~\cite{AlberFN04} who  gave the first linear-sized kernel  for the {\sc Dominating Set} problem
on planar graphs. The work of Alber et al.~\cite{AlberFN04} triggered an explosion
of papers on kernelization, and kernels of linear sizes were obtained for a variety of parameterized problems  on planar graphs including
\textsc{Connected
Vertex Cover, Minimum Edge Dominating Set, Maximum Triangle
Packing, Efficient Edge Dominating Set, Induced Matching,
Full-Degree Spanning Tree, Feedback Vertex Set, Cycle Packing}, and
\textsc{Connected Dominating Set}
\cite{AlberFN04,BodlaenderP08,BodlaenderPT08,ChenFKX07,GuoNW06,KanjPXS08,LokshtanovMS09,MoserS07}.
 We refer to the surveys~\cite{GuoN07,kratsch2014EATCS,MisraRS11} as well as the recent textbook  \cite{Cygan15_book} for a detailed treatment of the area of kernelization.
Since most of the problems known to have polynomial kernels on planar graphs are bidimensional, the existence of links between bidimensionality and kernelization was conjectured and left as an open problem  in~\cite{DemaineFHT05jacm}.


\medskip  In this work we show that  every  bidimensional problem with a simple separation property, which is a weaker property than the one required in the framework of Demaine and Hajiaghayi for PTASs  \cite{DemaineHaj05}  and which is expressible in the language of Counting Monadic Second Order Logic (CMSO) (we postpone these definitions till the next section) has a linear kernel on planar and even much more general classes of graphs. In this paper all the problems
are parameterized by the {solution size}. Our main result is the following meta-algorithmic result.
%

\begin{theorem}\label{thm:main_result_bidim}
Every CMSO-definable linear-separable minor-bidimensional problem $\Pi$ admits  a linear kernel on graphs excluding some fixed graph as a minor.
Every CMSO-definable linear-separable contraction-bidimensional problem $\Pi$ admits  a linear kernel on graphs excluding some fixed apex graph as a minor.
\end{theorem}

Theorem~\ref{thm:main_result_bidim} implies the existence of linear kernels for many parameterized problems on apex-minor-free or minor-free graphs. For example, it implies that 
\textsc{Treewidth-$\eta$-Modulator}, which is to decide whether an input  graph can be turned into a graph of treewidth at most $\eta$, by removing at most $k$ vertices, admits a kernel of size $\cO(k)$ in $H$-minor-free graphs. Other applications of this theorem are linear kernels for   \probrDS, \probCDS, \probCVC,  \probIS, or \probrSC on apex-minor-free graphs and for  \probCycPacking and \probFVS on  minor-free graphs, as well as for many other packing and covering problems. 
For many of these  problems these are the first polynomial kernels on such classes of graphs. 

\paragraph{High level overview of the main proof ideas.}

Our approach is built on the work of Bodlaender et al.~\cite{BodlaenderFLPST16meta} who proved the first meta-theorems on kernelization. 
The results in \cite{BodlaenderFLPST16meta} imply that every parameterized problem that has finite integer index and satisfies an additional surface-dependent property, called a quasi-coverability property, has a linear kernel on graphs of bounded genus.

The kernelization framework in \cite{BodlaenderFLPST16meta} is based on the following idea.
Suppose that  every \yesinstance\ of a given parameterized problem admits a protrusion decomposition.  In other words, suppose that the vertex set of an input graph $G$ can be partitioned in 
sets $R_0, R_1, \dots, R_\ell$, where $|R_0|$ and $\ell$ are of size linear in parameter $k$, $R_0$ separates $R_i$ and $R_j$, for $1\leq i<j\leq \ell$, and every  set $R_i$, $i\in\{1,\ldots,\ell\}$, is a protrusion, i.e. induces a graph of constant treewidth with a constant number of neighbours in $R_0$. 
Then the kernelization algorithm uses only one reduction rule, which is based on  finite integer index  properties of the problem in question, and replaces a protrusion with a protrusion of constant size.
 
%
%
%

In this paper we use  exactly the same reduction rule as the one used in the kernelization algorithm given in \cite{BodlaenderFLPST16meta} for obtaining kernels on planar graphs and graphs of bounded genus.  
The novel technical contribution of this paper is twofold.  First, we introduce a new 
  way the kernel sizes for bidimensional problems are analyzed.  The analysis of kernel sizes in~ \cite{BodlaenderFLPST16meta} requires ``topological" decompositions of the given graph, in the sense that the partitioning of the graph into regions with small border, or protrusions, strongly depends on the embedding of the graph into a surface. Then topological properties of the embedding are used to prove the existence of  a protrusion decomposition. 
While such an approach works well when we have a topological embedding it seems difficult to extend it to graphs excluding some fixed graph as a minor.
Instead of taking the topological approach, we apply bidimensionality and suitable variants of the
 Excluded Grid Theorem~\cite{DemaineH08,FominGT11cont}. Roughly speaking, we show that bidimensionality and separability implies the existence of a protrusion decomposition. 
This makes our arguments not only much more general but also considerably simpler than the analysis in~\cite{BodlaenderFLPST16meta}. Our second technical contribution is the proof that 
every CMSO-definiable separable problem has a finite integer index.
Pipelined with the framework from \cite{BodlaenderFLPST16meta}, these results imply the proof of Theorem~\ref{thm:main_result_bidim}.

The remaining part of the  paper is organized as follows. In Section~\ref{sec:prelims}, we provide definitions and notations used in the remaining part of the paper. 
 Section~\ref{sec:decprotrusions}  is devoted to combinatorial lemmata on separation properties of bidimensional problems. Based on these combinatorial lemmata, we prove a novel decomposition theorem (Theorem~\ref{thm:protrusiondecomp}),  which is  the first main technical  contribution of this work.  In Section~\ref{sec:FII}, we prove the second main technical contribution of the paper (Theorem~\ref{fiiopoiok})  about the finite integer index of  separable CMSO-optimization problems. 
 In Section~\ref{sec:puttin}, we prove the main result of the paper about linear kernels.
In the  concluding  Section~\ref{sec:conclusion} we discuss the connection between   the separable-bidimensional property of a problem and  the quasi-coverable property, which was used in   the meta-theorem  from \cite{BodlaenderFLPST16meta}. 

\paragraph{Relevant results}
Let us  provide a brief overview of the relevant results appeared since 2010, when the conference version of this paper was published. 
 The properties of {SQGM} and {SQGC} graph classes were used to design approximation, FPT, counting,  and kernelization algorithms on various graph classes in 
 \cite{FominLS18excl,FominLS18excl,GrigorievKT14bidi,BasteT17cont,KimST18data}. The issues of constructiveness of the kernelization algorithms provided in this paper is discussed by Garnero et al. in~\cite{GarneroPST15,Garnero2018}. Extension of some of our results to graphs excluding a topological minor is given by Kim et al.  \cite{KimLPRRSS16} (see also~\cite{KimST18data} for recent applications to counting problems). For \textsc{Dominating Set} or \textsc{Connected Dominating Set}, linear kernels obtained in this paper for apex-minor free were extended to much more general classes: 
  $H$-topological-minor-free graphs~\cite{FominLST18kern}.  For  \textsc{Dominating Set} linear kernels were obtained even for more general classes of graphs like graphs of bounded expansion \cite{Drange16}, see also \cite{EickmeyerGKKPRS17neig} for even more general results. Finally, see  \cite{GiannopoulouPRT17line} for kernelization
   results when excluding graphs under other partial ordering relations, different than minors.

%

\section{Preliminaries}\label{sec:prelims}

In this section we give various definitions  used in the paper.
We use $\Bbb{N}$ to denote the set of all non-negative integers and $\Bbb{Z}$ to denote the set of all integers.

\paragraph{Concepts from Graph Theory}
 Let~$G$ be a graph.  We use the notation $V(G)$ and $E(G)$ for the vertex set and  the edge set of $G$ respectively.  We say that a graph $H$ is a {\em subgraph} of $G$ if $V(H)\subseteq V(G)$ and $E(H)\subseteq E(G).$ Given a set $S\subseteq V(G)$ we denote by $G[S]$ the subgraph $G'$
of $G$ where $V(G')=S$ and $E(G')=\{xy \in E(G)\mid \{x,y\} \subseteq S\}$ and we call $G'$ the {\em subgraph of $G$ induced by $S$} or we simply say that $G'$ is an induced subgraph of $G.$

For every $S \subseteq V,$ we denote by  $G - S$ the graph obtained  from $G$ by removing the vertices of $S,$ i.e.  $G- S=G[V(G)\setminus S]$.  For vertex $v\in V(G)$ we also use $G-v$ for $G-\{v\}$.
For a set  $S \subseteq V(G)$, we define $N_G(S)$ to be the \emph{open neighborhood} of $S$ in $G$, which is the set of vertices from $V(G)\setminus S$ adjacent to vertices of $S$. The  \emph{closed neighborhood} of $S$ is $N_G[S]:=N(S)\cup S$.
Given a set $S\subseteq V(G)$,  we denote by $\partial_{G}(S)$ the set of all vertices in $S$ that are adjacent in $G$ with vertices not in $S$. Thus  $N_{G}(S)=\partial_{G}(V(G)\setminus S)$.
 
\paragraph{\bf Treewidth.}
A {\em tree decomposition} of a graph $G$ is a pair $\mathcal{T}=(T,\{X_i\}_{i\in V(T)})$, where $T$ is a tree whose every node $i$ is assigned a vertex subset $X_i\subseteq V(G)$, called a bag,
such that the following three conditions hold:
\begin{description}
\item[(T1)] $\bigcup_{i\in V(T)} X_i =V(G)$. In other words, every vertex of $G$ is in at least one bag.
\item[(T2)] For every $uv\in E(G)$, there exists a node $i$ of $T$ such that bag $X_i$ contains both $u$ and $v$.
\item[(T3)] For every  $u\in V(G)$, the set $T_u = \{i\in V(T) \mid  u\in X_i\}$  (i.e., the set of nodes whose corresponding  bags contain $u$) induces a subtree of $T$.
\end{description}
The {\em width} of a tree decomposition $\mathcal{T}=(T,\{X_i\}_{i\in V(T)})$  equals $\max_{i\in V(T)} |X_i| - 1$, that is, the maximum bag size, minus one. The {{\em treewidth}} of a graph $G$, denoted by $\tw(G)$, is the minimum possible width of a tree decomposition of $G$. To distinguish between the vertices of the decomposition tree $T$ and the vertices of the graph $G$, we will refer to the vertices of $T$ as {\em{nodes}}. Treewidth is can be seen as a measure of the topological resemblance of a graph to the structure of a tree. It has been used by Robertson and Seymour, in~\cite{RobertsonS3}, as a cornerstone parameter of Graph Minors  and its trace as a parameter goes back to the early 70s~\cite{Bertele72nons,Halin76sfun,Gavril74thei}.

\paragraph{Separators and separations}
Let $G$ be  a graph, $Q\subseteq V(G)$, and let $A_{1},A_{2}\subseteq V(G)$ such that $A_1 \cup A_2 = V(G)$. We say that the pair $(A_{1},A_{2})$ is a {\em separation} of $G$ if there is no edge with one endpoint in $A_{1}\setminus A_{2}$ and the other in $A_{2}\setminus A_{1}$. The {\em order} of a separation $(A_1,A_2)$ is $|A_1  \cap A_2|$. For a vertex subset $Q \subseteq V(G)$ we say that a separation  $(A_1,A_2)$ is a {\em $2/3$-balanced separation} of $(G,Q)$ if each of the parts $A_1 \setminus A_2$ and $A_2 \setminus A_1$ contains at most $\frac{2}{3}|Q|$  vertices of $Q$. Balanced separators have been extensively studied in the context of graph algorithms (see e.g.~\cite{LiptonT79,AST90,LiptonT80,AlonST94}).

The following separation property of graphs of small treewidth is well known, see e.g.  \cite{Bodlaender98,Cygan15_book}.

\begin{proposition}
\label{lemma:balsep22}
Let $G$ be a graph and let $S\subseteq V(G).$ There is a $2/3$-balanced separation $(A_1,A_2)$ of $(G,S)$ of order at most $\tw(G)+1.$ 
\end{proposition}

%

The following is an easy fact about treewidth, see e.g. \cite{Bodlaender98}.
\begin{proposition}
\label{tcompo}
The treewidth of a graph is the maximum treewidth of its connected components.
\end{proposition}

\paragraph{Minors and contractions}
Given an edge  $e=xy$ of a graph $G,$ the graph  $G/e$ is obtained from  $G$ by contracting the edge $e,$ that is, the endpoints $x$  and $y$ are replaced by a new vertex $v_{x,y}$ which  is  adjacent to the old neighbors of $x$ and $y$ (except
from $x$ and $y$).  A graph $H$ obtained by a sequence of edge-contractions is said to be a \emph{contraction} of $G.$  We denote it by $H\leq_{c} G$. A graph $H$ is a {\em minor} of a graph $G$ if $H$ is the contraction of some subgraph
of $G$ and we denote it by $H\leq_{m} G$. We say that a graph $G$ is {\em $H$-minor-free} when it does not contain $H$ as a minor. We also say that a graph class ${\cal G}$ is {\em $H$-minor-free} (or, excludes $H$ as a minor)  when
all its members are $H$-minor-free.
A graph $G$ is an \emph{apex graph} if there exists a vertex $v$ such that $G- v$ is planar. A graph class ${\cal G}$ is \emph{apex-minor-free} if there exists an apex graph $H$ such that  ${\cal G}$ is $H$-minor free.

A graph class ${\cal G}$ is said to be {\em subgraph-closed} ({\em minor-closed/contraction-closed}) if every subgraph (minor/contraction) of a graph in ${\cal G}$ also belongs to ${\cal G}$. 

\paragraph{Grids and triangulated grids.}
Given a $k\in\Bbb{N}$, we denote by $\boxplus_{k}$ the $(k\times k)$-grid that is  the graph with vertex set $\{(x,y) \mid  x,y \in\{1,\dots, t\}\}$ and where two different vertices $(x,y)$ and $(x',y')$ are adjacent if and only if $|x-x'|+|y-y'| = 1$. Notice that $\boxplus_k$ has exactly $k^2$ vertices.

For  $k\in\Bbb{N}$, the graph $\Gamma_k$ is obtained from the grid $\boxplus_k$ by adding, for all $1 \leq x,y \leq k-1$, the edge with endpoints $(x+1,y)$ and $(x,y+1)$ and additionally making vertex $(k,k)$ adjacent to all the other vertices $(x,y)$ with $x \in \{1,k\}$ or $y \in \{1,k\}$, i.e., to the whole perimetric border  of $\boxplus_k$. Graph $\Gamma_9$ is shown in Fig.~\ref{fig-gamma-reg}. The graph $\Gamma_{k}$ has been defined in~\cite{FominGT11cont} in the context of bidimensionality theory.

\begin{figure}
 \begin{center}
\scalebox{.6}{\includegraphics{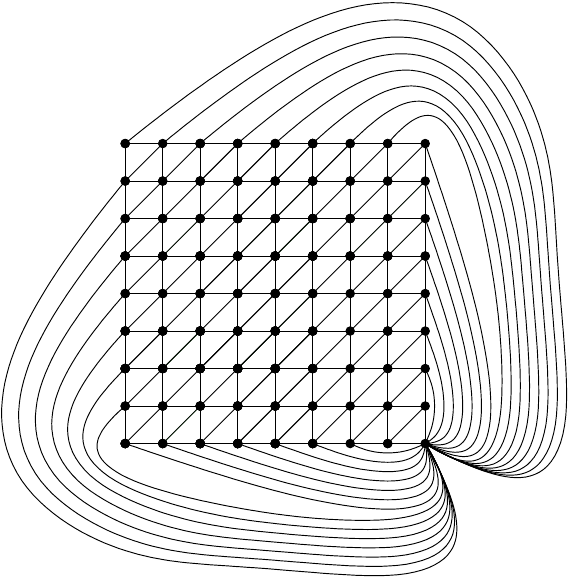}}
\caption{Graph $\Gamma_{9}$.}\label{fig-gamma-reg}
 \end{center}
 \end{figure}

 We also need the following result of Robertson and Seymour  \cite{RobertsonS3}.
 \begin{proposition}\label{prop:tw_grid}
 For every $k\geq 0$, $\tw(\boxplus_k)=k$.
 \end{proposition}

\paragraph{Parameterized graph problems.}
In general, 
a parameterized graph problem  $\Pi$  can be seen as a subset of $\Sigma^{*}\times \Bbb{N}$
where, in each instance $(x,k)$ of $\Pi,$ $x$ encodes a graph and $k$ is the parameter. In this paper we use an extension of this definition used in~\cite{BodlaenderFLPST16meta}  that permits the parameter $k$ to be negative 
with the additional constraint that either all pairs with non-positive value of the parameter 
are in $\Pi$ or that no such pair is in $\Pi$. Formally, a parameterized problem $\Pi$
is a subset of $\Sigma^{*}\times \Bbb{Z}$ where for all $(x_{1},k_{1}),(x_{2},k_{2})\in\Sigma^{*}\times \Bbb{Z}$
with $k_{1},k_{2}<0$ it holds that $(x_{1},k_{1})\in\Pi$ if and only if  $(x_{2},k_{2})\in\Pi$.
This extended definition encompasses the traditional one and is being adopted for technical reasons  
(see Subsection~\ref{subsec:finiinteginde}).
In an instance of a parameterized problem $(x,k),$ the integer $k$ is called the parameter.

%
The notion of {\em kernelization} is due to Downey and Fellows  \cite{DowneyF99}. Kernelization is formally defined as follows.

\begin{definition}[{\bf Kernelization}]
A {\em{kernelization algorithm}}, or simply a {\em{kernel}}, for a parameterized problem $\Pi$ is an algorithm $\mathcal{A}$ that, given an instance $(x,k)$ of $\Pi$, works in polynomial-time and returns an equivalent instance $(x',k')$ of $\Pi$. Moreover, 
there exists a computable function $g(\cdot)$ such that whenever $(x',k')$ is the output for an instance $(x,k)$, then it holds that $|x'|+k'\leq g(k)$. If the upper bound $g(\cdot)$ is a polynomial (linear) function of the parameter, then we say that $\Pi$ admits a {\em{polynomial (linear) kernel}}.
\end{definition}

We often abuse the notation and call the output of a kernelization algorithm, the ``reduced'' equivalent instance, also a kernel.


\paragraph{Bidimensionality}\label{chapter:logicpreliminar}

 Bidimensionality theory was  introduced by Demaine et al. in~\cite{DemaineFHT05jacm}. Here we closely follow the presentation of bidimensionality from the book of Cygan et al.  \cite{Cygan15_book}.

We will restrict our attention to vertex or edge subset problems.
 A {\em vertex subset problem} $\Pi$ is a parameterized problem where input is a graph $G$ and an integer $k$, the parameter is $k$. An instance $(G,k)$ is a \yesinstance if and only if there exists a set $S \subseteq V(G)$ such that $|S| \leq k$ for minimization problems (or $|S| \geq k$ for maximization problems) so that a predicate $\phi(G, S)$ is true. Here $\phi$ can be any computable function 
 which takes as input a graph $G$ and set $S \subseteq V(G)$ and outputs {\bf true} or {\bf false}. In the rest 
 of this paper, we always assume that a problem $\Pi$ is generated by the choice 
 of such a predicate $\phi$ and the choice of whether it is a minimization or a maximization problem. We will refer to $\phi$ as the \emph{feasibility function} of $\Pi$.

The interpretation is that $\phi$ defines the space of {\em feasible solutions} $S$ for a graph $G$ by returning whether $S$ is feasible for $G$. For an example, for  the {\sc Dominating Set} problem we have that $\phi(G,S) = {\bf true}$ if and only if $N[S] = V(G)$. {\em Edge subset problems} are defined similarly, with the only difference being that $S$ is a subset of $E(G)$ rather than $V(G)$.

Let us remark that there are many vertex/edge subset problems  which, at a first glance, do not look as if they could be captured by this definition. An example is the {\sc Cycle Packing} problem. Here the input is a graph $G$ and integer $k$, and the task is to determine whether $G$ contains $k$ pairwise vertex-disjoint cycles $C_1, C_2, \ldots , C_k$. This is a vertex subset problem because $G$ has $k$ vertex-disjoint cycles if and only if there exists a set $S \subseteq V(G)$  of size at least $k$ and $\phi(G,S)$ is
true, where $\phi(G,S)$ is defined as follows.
\begin{align*}
\phi(G,S) =\ & \exists \mbox{ subgraph } G'   \mbox{ of } G  \mbox{ such that }  \\
& \bullet \mbox{each connected component of } G'  \mbox{ is a cycle,} \\
& \bullet \mbox{and each connected component of } G' \mbox{ contains exactly one vertex of $S$.} 
\end{align*}

This definition may seem a bit silly, since checking whether $\phi(G,S)$ is true for a given graph $G$ and set $S$ is {\sf NP}-complete. In fact this problem is considered as a more difficult problem than {\sc Cycle Packing}. Nevertheless, this definition shows that {\sc Cycle Packing} is a vertex subset problem, which will allow us to give a linear kernel for {\sc Cycle Packing} on minor free graphs (see also~\cite{Garnero2018} for explicit bounds on linear kernels for packing problems).

For any vertex or edge subset minimization problem $\Pi$ we have that $(G,k) \in \Pi$ implies that  $(G,k') \in \Pi$ for all $k' \geq k$. Similarly, for a vertex or edge subset maximization problem we have that $(G,k) \in \Pi$ implies that  $(G,k') \in \Pi$ for all $k' \leq k$. Thus the notion of ``optimality'' is well defined for vertex and edge subset problems.
\begin{definition}[{\bf Optimum solution}]
For a vertex/edge subset minimization problem $\Pi$,  we define
$$OPT_\Pi(G) = \min\left\{k \mid (G,k) \in \Pi\right\}.$$
If no $k$ such that $(G,k) \in \Pi$ exists, $OPT_\Pi(G)$ returns $+\infty$. 
For a vertex/edge subset maximization problem $\Pi$, 
$$OPT_\Pi(G) = \max\left\{k \mid (G,k) \in \Pi\right\}.$$
If no $k$ such that $(G,k) \in \Pi$ exists, $OPT_\Pi(G)$ returns $-\infty$.
We say that a vertex (edge) set $S$ is  an \emph{optimum solution} of $\Pi$ for $G$ if $\phi(G,S) = \mbox{\bf true}$
and $|S|=OPT_\Pi(G)$.
We define $SOL_\Pi(G)$ to be a function that given as an input a graph $G$ returns an optimum solution  $S$  and returns {\bf null} if no such set $S$ exists (if many optimum solutions exist, we pick one arbitrarily).
\end{definition}

For many problems it holds that contracting an edge can not increase the size of the optimal solution. We will say that such problems are  {contraction closed}. Formally we have the following definition.
\begin{definition}[{\bf Contraction-closed problem}]
A vertex/edge subset problem $\Pi$ is {\em contraction-closed} if for any $G$ and $uv \in E(G)$, $OPT_\Pi(G/uv) \leq OPT_\Pi(G)$.
\end{definition}
If contracting edges, deleting edges and deleting vertices can not increase the size of the optimal solution, we say that the problem is minor-closed.
\begin{definition}[{\bf Minor-closed problem}]
A vertex/edge subset problem $\Pi$ is {\em minor-closed} if for any $G$, edge $uv \in E(G)$ and vertex $w \in V(G)$, $OPT_\Pi(G/uv) \leq OPT_\Pi(G)$, $OPT_\Pi(G \setminus uv) \leq OPT_\Pi(G)$, and $OPT_\Pi(G- w) \leq OPT_\Pi(G)$.

\end{definition}
The following  slight modification
of $OPT_{\Pi}$  makes possible to avoid dealing with 
asymptotic inequalities in the rest of the paper.  
\begin{definition}[{\bf Modified OPT}]
We also define $OPT_{\Pi}^{*}(G)=\max\{OPT_\Pi(G),1\}$.
\end{definition}

We are now ready to give the definition of bidimensional problems.

\begin{definition}[{\bf Bidimensional problem}]
\label{defdbidi}
A vertex/edge subset problem $\Pi$ is 
\begin{itemize}
\item {\em contraction-bidimensional}: if it is contraction-closed and there exists a positive real constant $\constbidimens$ such that  for every $ k\in\Bbb{N}$, $OPT_\Pi^{*}(\Gamma_k) \geq \constbidimens \cdot k^2$.
\item {\em minor-bidimensional}: if it is minor-closed and 
there exists a positive real constant $\constbidimens$ such that  for every $ k\in\Bbb{N}$, $OPT_\Pi^{*}(\boxplus_k) \geq \constbidimens \cdot k^2$.
%
%
\end{itemize}
\end{definition}


%

It is usually quite easy to determine whether a problem is contraction  (or minor)-bidimensional. Take for an example {\sc Independent Set}. Contracting an edge may never increase the size of the maximum independent set, so the problem is contraction closed. Furthermore, $\Gamma_k$
is a planar graph and because of the Four Color theorem, it contains an independent
set of size at least $\lfloor \frac{|V(\Gamma_{k})}{4}\rfloor$. %
Thus  {\sc Independent Set} is contraction-bidimensional. On the other hand, deleting edges may increase the size of a maximum size independent set in $G$. Thus {\sc Independent Set} is not minor-bidimensional. 


We would also like to comment why we define  bidimensionality by making use of $OPT_\Pi^{*}$ and not $OPT_\Pi$. The reason is that for most interesting problems like {\sc Vertex Cover} or {\sc Feedback Vertex Set}, the values of   $OPT_\Pi (\boxplus_k)$ on very small grids (like $\boxplus_1$ for {\sc Vertex Cover} and  $\boxplus_2$ for {\sc Feedback Vertex Set}) is zero. The definition of $OPT_\Pi^{*}$ takes care of such degenerate situations. 

An alternative way to define bidimensionality, by making use of  $OPT_\Pi$, would be to ask for the existence of a positive $\constbidimens$ such that that  $OPT_\Pi(\boxplus_k)$ (or $OPT_\Pi(\Gamma_k)$) is {\sl asymptotically bigger} than $\constbidimens \cdot k^2$, i.e.,  there exists   an integer $k_0>0$ such that  for every $ k\geq k_0$, $OPT_\Pi(\boxplus_k) \geq \constbidimens \cdot k^2$. This new definition is equivalent to the one of Definition~\ref{defdbidi}.
To see this, in the non-trivial direction, we set $k_1=\max\{k_{0},\sqrt{1/\constbidimens}\}$ 
which  implies that $\forall k\geq k_1\ OPT_{\Pi}(\boxplus_{k})=OPT_{\Pi}^*(\boxplus_{k})\geq \constbidimens \cdot k^2$. 
We now set $\constbidimens'=\min\{\constbidimens,1/\sqrt{k_{1}}\}$ and observe that 
if $0\leq k\leq k_1$, then $OPT^*_{\Pi}(\boxplus_k)\geq 1\geq \constbidimens'\cdot k^{2}$, while if $k\geq k_1$, then 
$OPT_{\Pi}^*(\boxplus_{k})=OPT_{\Pi}(\boxplus_{k})\geq \constbidimens'\cdot k^2$, as required. 

In this paper we adopted the definition of Bidimensionality that uses $OPT_\Pi^{*}$ because this makes our proofs easier to present.

\paragraph{Counting Monadic Second Order Logic}
\label{subsec:counmonasecoordelogi}

The syntax of Monadic Second Order Logic (MSO) of graphs includes the logical connectives $\vee,$ $\land,$ $\neg,$ 
$\Leftrightarrow ,$  $\Rightarrow,$ variables for 
vertices, edges, sets of vertices, and sets of edges, the quantifiers $\forall,$ $\exists$ that can be applied 
to these variables, and the following five binary relations: 
\begin{enumerate}

\item 
$u\in U$ where $u$ is a vertex variable 
and $U$ is a vertex set variable; 
\item 
 $d \in D$ where $d$ is an edge variable and $D$ is an edge 
set variable;
\item 
 $\mathbf{inc}(d,u),$ where $d$ is an edge variable,  $u$ is a vertex variable, and the interpretation 
is that the edge $d$ is incident with the vertex $u$; 
\item 
 $\mathbf{adj}(u,v),$ where  $u$ and $v$ are 
vertex variables  and the interpretation is that $u$ and $v$ are adjacent; \item 
 equality of variables representing vertices, edges, sets of vertices, and sets of edges.
\end{enumerate}

In addition to the usual features of monadic second-order logic, if we have atomic sentences testing whether the cardinality of a set is equal to $q$ modulo $r,$ where $q$ and $r$ are integers such that $0\leq q<r$ and $r\geq 2$, then 
this extension of the MSO is called {\em counting monadic second-order logic}. Thus CMSO is MSO enriched with the following atomic sentence for a set $S$: 
\begin{quote}
$\mathbf{card}_{q,r}(S) = \mathbf{true}$ if and only if $|S| \equiv q \pmod r.$ 
\end{quote}
For a detailed introduction on CMSO, see~\cite{ArnborgLS91,Courcelle90,Courcelle97}. 

\medskip

We consider CMSO sentences evaluated either on graphs or on annotated graphs. 
In this paper by annotated graph we mean a pair $(G,S)$, where $G$ is a graph and $S$ is either a vertex or edge subset of $G$. We remark that a graph can be annotated with more than one set and all of the considerations in this paper could be generalized in this sense. For simplicity, we omit this in this paper. 
A class of graphs ${\cal G}$ is {\em CMSO-definable} if there is a CMSO sentence on graphs  $\phi$ such that $G \in {\cal G}$ if and only if $G\models \phi$. Similarly a predicate $\phi$ on annotated graphs is {\em CMSO-definable} if there is a CMSO sentence $\consteps$ on annotated graphs   such that $\phi(G,S) = {\bf true}$ if and only if $(G,S) \models \consteps$.

A vertex/edge subset minimization (or maximization) problem with feasibility function $\phi$ is a {\sc min-CMSO} problem  (or {\sc max-CMSO} problem) if $\phi$ it is CMSO-definable.


\paragraph{Subquadratic grid minor/contraction property.} In general, 
it is known that there exists a constant $c$ such that any graph $G$ which excludes a $\boxplus_{k}$ as a minor has treewidth at most $O(k^c)$.  The exact value of $c$ remains unknown, but it is more than $2$ (see e.g.,~\cite{Thilikos12grap})
and at most $20$ 
(see~\cite{ChekuriC13,Chuzhoy15excl,Chuzhoy16}). We will restrict our attention to graph classes on which $c < 2$. In particular, we say that a graph class ${\cal G}$ has the {\em subquadratic grid minor property} (SQGM property, for short) if there exist constants $\lambda > 0$ and $1 \leq c < 2$ such that every graph $G \in {\cal G}$ which excludes    $\boxplus_{k}$ as a minor has treewidth at most $\lambda k^c$.

Problems that are contraction-closed but not minor closed are considered on more restricted classes of graphs. We say that a graph class ${\cal G}$ has the {\em subquadratic graph contraction property} (SQGC property for short) if there exist constants $\lambda > 0$ and $1 \leq c < 2$ such that any {\sl connected} graph $G \in {\cal G}$ which excludes a  $\Gamma_{k}$ as a contraction has treewidth at most $\lambda k^c$. 

%
\begin{observation}\label{ovsSQGCSQGM} Every graph class ${\cal G}$ with the SQGC property has the SQGM property. \end{observation}

\begin{proof}
Suppose that the graph class ${\cal G}$ has the SQGC property. This means that there is 
a $c\in[1,2)$ such that  if $G$ is a connected graph in ${\cal G}$ that cannot be contracted to  
$\Gamma_{k}$, then $\tw(G)\leq \lambda k^c$.
Suppose that $G$ is any graph in ${\cal G}$ excluding  $\boxplus_{k}$ as a minor.
Clearly all connected components of $G$ exclude  $\boxplus_{k}$ as a minor
and, as  $\boxplus_{k}$ is a minor of $\Gamma_{k}$ (in fact it is a subgraph),  the connected components of $G$ also exclude  $\Gamma_{k}$ as a contraction. This  implies that all connected components 
of $G$ have treewidth at most $\lambda k^c$, therefore $G$ has treewidth at most $\lambda k^c$, as required.
\end{proof}

The following proposition  follows directly from the linearity of excluded grid-minor in $H$-minor-free graphs proven by  Demaine and Hajiaghayi~\cite{DemaineH08}  and its analog for contraction-minors from \cite{FominGT11cont}.
\begin{proposition}\label{prop:linear_grid_minor}
For every graph $H$, {$H$-minor-free} graph class ${\cal G}$ has the SQGM property with $c=1$.  If $H$ is an apex graph, then ${\cal G}$ has the SQGC property  with $c=1$.
\end{proposition}

\paragraph{Problem restrictions.}
We say that a parameterized problem $\Pi$ is a problem on the graph class ${\cal G}$ if every yes-instance $(G,k)$ of $\Pi$ satisfies $G \in {\cal G}$. The {\em restriction} of a parameterized problem $\Pi$ to a graph class ${\cal G}$, denoted by $\Pi \doublecap {\cal G}$ is defined as follows.
$$ {\rm \Pi} \doublecap {\cal G} = \{(G,k)\mid (G,k)\in \Pi \mbox{~and~}G \in {\cal G}\}. $$
For a parameterized problem $\Pi$ (on general graphs) we will refer to the restriction of $\Pi$ to ${\cal G}$ by ``$\Pi$ on ${\cal G}$''.


\section{Decomposing into protrusions}\label{sec:decprotrusions}
In this section we give the main technical combinatorial contribution of this work establishing a protrusion decomposition theorem for linearly-separable bidimensional problems. 
The proof of the theorem is done in two steps. First we show that every
 graph class with the SQGM or the SQGC property admit a treewidth-modulator of size linear in the parameter of a 
linearly-separable (minor- or contraction-) bidimensional problem 
(Subsection~\ref{subsec:seps_and_tw_moduls}). Then in Subsection~\ref{subsec:protr_decomps}, we show that graph classes with SQGM and SQGC properties and having linear treewidth-modulators, 
can be decomposed into protrusions.

\subsection{Parameter-treewidth bounds}
The following lemmata establish {\em parameter-treewidth bounds}, that is tight relationships between the size of the optimal solution and the treewidth of the input graph. This relationship was first observed by Demaine et al. \cite{DemaineFHT05jacm}. The bound for 
contraction-bidimensional problems presented here is essentially identical to the one presented in Fomin et al.~\cite{FominGT11cont}. We re-prove the lemmata here because of slight differences in definitions.
 
\begin{lemma}\label{lem:bidimParameterTreewidth}
For any minor-bidimensional problem $\Pi$ on a graph class ${\cal G}$ with the SQGM property, there exist constants $0 < \alpha$, $\frac{1}{2}\leq  \constPower < 1$ such that for any graph $G \in {\cal G}$, $\tw(G) \leq \alpha \cdot  (OPT^{*}_\Pi(G))^{\constPower}$. 
\end{lemma}

\begin{proof}
Let $\lambda>0$ and $1\leq c<2$ be the constants from the definition of the SQGM property, that is any graph $G \in {\cal G}$ which excludes a  $\boxplus_{k}$ as a minor has treewidth at most $\lambda k^c$.
Because of the minor-bidimensionality of $\Pi$, 
 there is   $\beta>0$ such that  $OPT^{*}_\Pi(\boxplus_k) \geq \beta k^2$ for every $k\in\Bbb{N}$. Consider now a graph $G \in {\cal G}$. 

Let $k$ be the maximum integer such that $G$ contains $\boxplus_k$ as a minor. This means that $G$ excludes $\boxplus_{k+1}$ as a minor, therefore $\tw(G) \leq  \lambda (k+1)^c$. 
Rearranging terms yields that $k\geq (\frac{\tw(G)}{\lambda})^{1/c}-1$. Since $\Pi$ is minor-closed, it follows that $OPT_\Pi(G) \geq OPT_\Pi(\boxplus_k)$, therefore 
$$OPT^*_\Pi(G) \geq OPT^*_\Pi(\boxplus_k) \geq \beta k^2 \geq \beta \big((\frac{\tw(G)}{\lambda})^{1/c}-1\big)^2,$$
hence
$$\tw(G)\leq \lambda \big(\big(\frac{OPT^*_\Pi(G)}{\beta}\big)^\frac{1}{2}+1\big)^c.$$
Recall that $OPT^*_\Pi(G) \geq 1$. We set $x=(\frac{OPT^*_\Pi(G)}{\beta})^\frac{1}{2}$ and, as $x\geq \frac{1}{\sqrt{\beta}}$, we obtain that $x+1\leq (1+\sqrt{\beta})x$.
This implies that $$\tw(G)\leq \lambda ((1+\sqrt{\beta})(\frac{OPT^*_\Pi(G)}{\beta})^\frac{1}{2})^c=\lambda\cdot (1+\frac{1}{\sqrt{\beta}})^{c}\cdot (OPT^*_{\Pi}(G))^{\frac{c}{2}}.$$
We can now set $\alpha=\lambda\cdot (1+\frac{1}{\sqrt{\beta}})^{c}$, $\constPower=c/2$ and observe that $\tw(G)\leq \alpha\cdot  (OPT^*_\Pi(G))^{\constPower}$.

Since $c < 2$ in the definition of the SQGM property, we obtain  that $\constPower<1$ and the statement of the lemma follows.
\end{proof}
\begin{lemma}\label{lem:contractionBidimParameterTreewidth}
For any contraction-bidimensional problem $\Pi$ on a graph class ${\cal G}$ with the SQGC property, there exist constants $0 < \alpha$, $\frac{1}{2}\leq  \constPower < 1$, such that for any connected graph $G \in {\cal G}$, $\tw(G) \leq \alpha \cdot (OPT^*_\Pi(G))^{\constPower}$. 
\end{lemma}

The proof of the Lemma~\ref{lem:contractionBidimParameterTreewidth} is almost identical to the proof of Lemma~\ref{lem:bidimParameterTreewidth}. The differences between Lemmata~\ref{lem:bidimParameterTreewidth} and~\ref{lem:contractionBidimParameterTreewidth} are as follows. Lemma~\ref{lem:bidimParameterTreewidth} is for minor-closed problems, on graph classes ${\cal G}$ with the SQGM property and works for {\em every} $G \in {\cal G}$. 
Lemma~\ref{lem:contractionBidimParameterTreewidth} is for contraction-closed problems, on graph classes ${\cal G}$ with the SQGC property and works for {\em connected} graphs $G \in {\cal G}$. The connectivity requirement here is necessary: For SQGC property we require that  there exist constants $\lambda > 0$ and $1 \leq c < 2$ such that any {\em connected} graph $G \in {\cal G}$ which excludes a  $\Gamma_{k}$ as a contraction has treewidth at most $\lambda k^c$.  Moreover, it is possible to provide   an example of a contraction-bidimensional problem $\Pi$ such that Lemma~\ref{lem:contractionBidimParameterTreewidth} does not hold for $\Pi$ and disconnected graphs,  see \cite[Exercise~7.42]{Cygan15_book}. However, as we will see soon, if in addition to contraction-bidimensionality the problem is separable, then the connectivity condition is not necessary anymore. 



\subsection{Separability and treewidth modulators}\label{subsec:seps_and_tw_moduls}
We now restrict our attention to problems $\Pi$ that are somewhat well-behaved in the sense that whenever we have a small separator in the graph that splits the graph in two parts $L$ and $R$, the intersection $|X \cap L|$ of $L$ with any optimal solution $X$ to the entire graph is a good estimate of $OPT_\Pi(G[L])$. This restriction allows us to prove decomposition theorems which are very useful for giving kernels. Similar decomposition theorems may also be used to give approximation schemes, see~\cite{DemaineHaj05,FominLRS10}. 
\begin{definition}[{\bf Separability}]
Let $f : \mathbb{N} \rightarrow  \mathbb{N}$ be a function. We say that a vertex subset problem $\Pi$ is $f$-{\em separable} if for every graph $G$,  every optimum solution $S$ of $\Pi$ for $G$,   and every subset $L \subseteq V(G)$ 
$$|S \cap L| - f(t) \leq OPT_\Pi(G[L]) \leq |S \cap L| + f(t)$$
where  $t=|\partial_G(L)|$. 
In the case where $\Pi$ is an edge  subset problem the same definition as above applies with the difference that we agree to interpret $S\cap L$ by $S\cap E(G[L])$.
The problem $\Pi$ is called {\em separable} if there exists a function $f$ such that $\Pi$ is $f$-{\em separable}. $\Pi$ is called {\em linear-separable} if function $f$ is linear. That is, there exists a constant $\sigma$ such that $\Pi$ is $\sigma \cdot t$-separable.
\end{definition}

 There are many problems that are  contraction (or minor)-bidimensional, linear-separable including  \probrDS, \probCDS, \probCVC, \probVC, \probIS, \probFVS,  and \probrSC, \probCycPacking as well as many packing and covering problems. Another important generic problem which is minor-bidimensional and linear-separable, is the 
 \textsc{Treewidth-$\eta$-Modulator} problem defined below. 
 We refer to  \cite{DemaineHaj05} and \cite[Section~4]{FominLS18excl} for definitions of these problems and the proofs that they are contraction-bidimensional and  linear-separable.

 Just as one example, let us consider \probCycPacking. It is clearly minor-bidimensional. 
It is easy to see that \probCycPacking is a minor-closed problem. Since grid $\boxplus_{k}$ contains $\Omega(k^2)$ vertex-disjoint cycles,  
\probCycPacking is a minor-bidimensional problem. For linear-separability, 
we can view  \probCycPacking as a problem of finding a maximum vertex set $X$ such that there is a subgraph $H$ of $G$, such that every connected component of $H$ contains exactly one vertex of $X$ and is a cycle. Observe that 
  by deleting $t$ vertices from graph $G$, we cannot hit more than 
 $t$ vertex-disjoint cycles from the solution. Hence the linear-separability of   \probCycPacking follows.

A nice feature of separable contraction-bidimensional problems is that it is possible to extend Lemma~\ref{lem:contractionBidimParameterTreewidth} to disconnected graphs. 
%
\begin{lemma}\label{lem:bidimParameterTreewidthSep}
For any contraction-bidimensional separable problem $\Pi$ on a graph class ${\cal G}$ with the SQGC property, there exist constants $0 < \alpha$ and  $\frac{1}{2}\leq  \constPower < 1$ such that for any graph $G \in {\cal G}$, $\tw(G) \leq \alpha \cdot (OPT^*_\Pi(G))^{\constPower}$. 
\end{lemma}

\begin{proof}
Let $S$ be any optimal solution for $\Pi$.
Since $\Pi$ is separable there exists a function $f : \mathbb{N} \rightarrow  \mathbb{N}$
 such that for every graph $G$ and every connected component $C$ of $G$, it holds that 
$$OPT_\Pi(G[C]) \leq |S \cap  C | + f(|\partial_{G}( C )|) \leq |S| + f(0) = OPT_\Pi(G) + f(0).$$
We set $c=f(0)$. The above implies that $OPT^*_\Pi(G[C]) \leq OPT^*_\Pi(G) + c.$
By Lemma~\ref{lem:contractionBidimParameterTreewidth}, there exist constants  $\alpha'>0$ and $\frac{1}{2}\leq  \constPower < 1$ such that $\tw(G[C]) \leq \alpha' \cdot \big(OPT^*_\Pi(C)\big)^\constPower$ for every connected component $C$. By Proposition~\ref{tcompo}, the treewidth of $G$ is at most the maximum of the treewidth of its connected components. Since 
  $OPT^{*}_{\Pi}(G)\geq 1$,  we have that 
 $$\tw(G)\leq \alpha' \cdot \big(OPT^*_\Pi(G) + c\big)^\constPower\leq \alpha' \cdot \big((c+1)\cdot OPT^*_\Pi(G)\big)^\constPower=  \alpha' \cdot (c+1)^\constPower\cdot \big(OPT^*_\Pi(G)\big)^\constPower.$$
The lemma follows if we set $\alpha= \alpha' \cdot (c+1)^\constPower$.
\end{proof}

We will now consider a sequence of ``canonical'' bidimensional problems, with one problem for every $\eta \geq 0$. We say that a set $S \subseteq V(G)$ is a {\em \twetamod } if $\tw(G- S) \leq \eta$.  We define the following problem. 

\begin{center}
\begin{boxedminipage}{.96\textwidth}
\textsc{Treewidth-$\eta$-Modulator}\\
\begin{tabular}{ r l }
\textit{~~Instance:} & A graph $G$, and integer $k \geq 0$.\\
\textit{Parameter:} & $k$.\\
\textit{Problem:} & Decide whether there exists a  \twetamod   of $G$ of size at most $k$.
\end{tabular}
\end{boxedminipage}
\end{center}

Clearly, $\Pi=$\textsc{Treewidth-$\eta$-Modulator} is a minimization problem
and it is easy to see that it is also  minor-closed. By Proposition~\ref{prop:tw_grid},  every $(\eta+1) \times (\eta+1)$ subgrid of $\boxplus_k$, for $k\geq 1$, must contain at least one vertex of any solution, therefore 
\[OPT_\Pi^{*}(\boxplus_k) \geq  \lfloor\frac{k^2}{(\eta+1)^2}\rfloor.\]

Thus the problem is minor-bidimensional. 
\medskip


To see that the problem is separable,  consider any graph $G$, and vertex subset $L \subseteq V(G)$ with $\partial(L)  = t$. For any \twetamod   $S$ of $G$ we have that $(S \cap L) \cup \partial(L)$ is a  \twetamod   of $G[L]$. This proves $OPT_\Pi(G[L]) \leq |S \cap L| + t$. On the other hand, for any optimal \twetamod  $S$ of $G$ and \twetamod  $S_L$ of $G[L]$, we have that $|S \cap L| \leq |S_L| + t$, since otherwise $(S \setminus L) \cup S_L \cup \partial(L)$ is a \twetamod  of $G$ of size strictly smaller than $|S|$, contradicting the optimality of $S$. It follows that $|S\cap L| - t \leq OPT_\Pi(G[L])$. This shows  that the \textsc{Treewidth-$\eta$-Modulator} problem is minor-bidimensional and separable. 

The fact that the \textsc{Treewidth-$\eta$-Modulator}  problem is minor-bidimensional, together with Lemma~\ref{lem:bidimParameterTreewidth} yields the following observation.

\begin{observation}\label{obs:etaTransversalTreewidth}
For every graph class ${\cal G}$ with the SQGM property and every $\eta \geq 0$ there exist constants $\alpha >0$ and $\frac{1}{2}\leq  \constPower < 1$ such that every graph $G \in {\cal G}$  with a non-empty \twetamod  $S$ has treewidth at most $\alpha |S|^\constPower$.
\end{observation}

\begin{proof}
Let $\Pi=$\textsc{Treewidth-$\eta$-Modulator}, let $G\in{\cal G}$ and let $S$ be a \twetamod\ of $G$.
As $\Pi$ is minor-bidimensional and ${\cal G}$ has the  SQGM property,
by Lemma~\ref{lem:bidimParameterTreewidth}, there exists  $0 < \alpha$, $\frac{1}{2}\leq  \constPower < 1$ such that $G \in {\cal G}$, $\tw(G) \leq \alpha \cdot  (OPT^{*}_\Pi(G))^{\constPower}$.
If $OPT_{\Pi}(G)=0$, then 
$OPT_{\Pi}^{*}(G)=1$ which implies that $\tw(G)\leq \alpha\leq \alpha |S|^{\constPower}$ (recall that $|S|\geq 1$).
If $OPT_{\Pi}(G)\geq 1$, then $OPT_{\Pi}(G)=OPT_{\Pi}^{*}(G)$, therefore $\tw(G) \leq \alpha (OPT_\Pi(G))^{\constPower}\leq \alpha |S|^{\constPower}$.
\end{proof}

Next we show that the  {\sc Treewidth-$\eta$-Modulator} problems are canonical bidimensional problems in the following sense.

\begin{lemma}\label{lem:bidimTwModulator}
For any real $\consteps > 0$ and minor-bidimensional linear-separable problem $\Pi$ on graph class ${\cal G}$ with the SQGM property, there exists an integer $\eta \geq 0$ such that   every $G \in {\cal G}$ has a \twetamod  $S$ of size at most $\consteps \cdot OPT_\Pi(G).$
\end{lemma}

\begin{proof}
Let $\sigma$ be a constant such that $\Pi$ is $(\sigma \cdot t)$-separable. Let $\alpha'>0$ and $\frac{1}{2}\leq  \constPower < 1$ be the constants from Lemma~\ref{lem:bidimParameterTreewidth}. In particular $\tw(G) \leq \alpha' \cdot (OPT_\Pi(G))^\constPower$ if $OPT_\Pi(G)>0$ and 
$\tw(G) \leq \alpha'$ if $OPT_\Pi(G)=0$.
Set $\alpha = \max\{\alpha', 1\}$. Furthermore, if $\sigma < 1$ then $\Pi$ is $t$-separable, and so we may assume without loss of generality that $\sigma \geq 1$.

We now define a few constants.  The reason  these constants are defined the way they are will become clear during the course of the proof.  Finally we  set $\eta$ based on $\alpha$, $\sigma$, $\constPower$ and $\consteps$. 
\begin{itemize}
\item Set $\rho = \frac{1^\constPower + 2^\constPower - 3^\constPower}{3^\constPower}$ and note that $\rho > 0$.
\item set $\gamma = 4\alpha\sigma$,
\item set $\delta = \frac{\gamma(2\consteps + 1)}{\rho}$, 
\item set $k_0 = (3+3\gamma)^\frac{1}{1-\constPower} + 3\cdot (\frac{\delta}{\consteps})^\frac{1}{1-\constPower}$, 
%
%
(notice that $k_{0}>3$)
 and
\item set $\eta = \alpha \cdot k_0^\constPower$.
\end{itemize}

%
%

We prove, by induction on $k$, the following stronger statement.\medskip

\noindent{\em Claim:} For any integer $k \geq \frac{1}{3}k_0$, every graph $G \in {\cal G}$ such that $OPT_\Pi(G) \leq k$ has a \twetamod  of size at most $\consteps{k} - \delta k^\constPower$.

\noindent{\em Proof of Claim:} 
In the base case we consider any $k$ such that $\frac{1}{3}k_0 \leq k \leq k_0$.  By Lemma~\ref{lem:bidimParameterTreewidth}, any graph $G \in {\cal G}$ such that  $OPT_\Pi(G) \leq k $ has treewidth at most $\alpha \cdot k^\constPower \leq \eta$ (recall that $k\geq \frac{1}{3}k_0\geq 1$). Thus $G$ has a \twetamod  of size $0$.  Also by the definition of $k_0$, we have
that
$k_{0}\geq  3\cdot (\frac{\delta}{\consteps})^\frac{1}{1-\constPower}= (\frac{3\delta}{\consteps\cdot 3^{\constPower}})^{\frac{1}{1-\constPower}}$. Hence $ k_0^{1-\constPower}\geq (\frac{3\delta}{\consteps\cdot 3^{\constPower}})$. This yields that 
$\frac{3^\constPower}{\delta}k_0^{1-\constPower}\geq \frac{3}{\consteps}$, therefore $\delta\frac{1}{3^{\constPower}}\cdot k_0^{\constPower-1}\leq \frac{\consteps}{3}$, and finally, $\delta\frac{1}{3^{\constPower}}\cdot k_0^{\constPower}\leq \frac{\consteps}{3}k_0.$ Consequently,  
\begin{eqnarray} 0 \leq \frac{\consteps \cdot k_0}{3} - \delta \left(\frac{k_0}{3}\right)^\constPower.\label{eq:k02}\end{eqnarray}

%
%
%
%
%
%
%
By  \eqref{eq:k02}, we have that the size of  \twetamod of $G$, which is $0$, satisfies the following 
\[
0 \leq \consteps\frac{1}{3}k_0 - \delta \left(\frac{1}{3}k_0\right)^\constPower \leq \consteps k - \delta k^\constPower.
\] 
In the last inequality we used the fact that for any $\frac{1}{2} \leq \constPower < 1$, $\consteps$ and $\delta,$ the function $\consteps k - \delta k^\constPower$ is monotonically increasing from the first point where it becomes positive. This fact may easily be verified by differentiation. This completes the proof of the base case. 

\smallskip

For the inductive step, let $k > k_0$ and suppose that the statement is true for all values below $k$. We prove the statement for $k$. Consider a graph $G \in {\cal G}$ such that $OPT_\Pi(G) \leq k$. By Lemma~\ref{lem:bidimParameterTreewidth}, the treewidth of $G$ is at most $\tw(G) \leq \alpha \cdot k^\constPower$. By Proposition~\ref{lemma:balsep22} applied to $(G, SOL_\Pi(G))$, there is a 2/3-balanced separation $(A_1, A_2)$ of $(G, SOL_\Pi(G))$ of order at most $\tw(G) + 1 \leq \alpha \cdot k^\constPower + 1$. Let $L = A_1 \setminus A_2$, $S = A_1 \cap A_2$ and $R = A_2 \setminus A_1$. Note that there are no edges from $L$ to $R$. Since $(A_1, A_2)$ is a 2/3-balanced separation it follows that there exists a real $\frac{1}{3} \leq a \leq \frac{2}{3}$ such that $|L \cap SOL_\Pi(G)| \leq a|SOL_\Pi(G)|$ and  $|R \cap SOL_\Pi(G)| \leq (1-a)|SOL_\Pi(G)|$.
 
Consider now the graph $G[L \cup S]$. Since $L$ has no neighbors in $R$ (in $G$) and $\Pi$ is $(\sigma \cdot t)$-separable, it follows that 
\begin{align*}
OPT_\Pi(G[L \cup S]) & \leq |SOL_\Pi(G) \cap (L \cup S)| + \sigma|S| \\
& \leq ak + (\alpha k^\constPower + 1) + \sigma(\alpha k^\constPower + 1) \\
& \leq ak + (\alpha k^\constPower + 1)(\sigma + 1) \leq ak + \gamma k^\constPower.
\end{align*}
Here the last inequality follows from the assumption that $k \geq k_0 \geq 1$ and the choice of $\gamma$. 

We claim that for $k\geq k_0$, 
\begin{equation}\label{eq:k01}ak  + \gamma k ^\constPower \leq k  - 1.\end{equation}
By the choice of $k_0$ and the fact that $\constPower>0$, we have 
\begin{eqnarray}
k_0 \geq  (3+3\gamma)^\frac{1}{1-\constPower}\Rightarrow k_0 \geq (\frac{3}{k_0^\constPower}+3\gamma)^\frac{1}{1-\constPower} \Rightarrow \nonumber\\k_0^{1-\constPower} \geq \frac{3}{k_0^\constPower}+3\gamma\Rightarrow
 \frac{k_0}{3}\geq 1+\gamma k_{0}^{\constPower}\Rightarrow\nonumber\\ 
 \frac{2}{3}k_0 + \gamma k_0^\constPower \leq k_0 - 1\nonumber 
\end{eqnarray} 
%
Because  for $k\geq k_0$ the function $\frac{2}{3}k  + \gamma k ^\constPower - k  + 1$ decreases, which is easily  verified by differentiation. Hence $\frac{2}{3}k  + \gamma k ^\constPower  \leq k  - 1$  and because  $a\leq \frac{2}{3}$, \eqref{eq:k01} follows. 

   Further $ak + \gamma k^\constPower \geq \frac{1}{3}k_0$ since $a \geq \frac{1}{3}$. Thus we may apply the induction hypothesis to $G[L \cup S]$ and obtain a \twetamod  $Z_L$ of $G[L \cup S]$, such that
\begin{align*}
|Z_L| & \leq \consteps(ak + \gamma k^\constPower) - \delta\left(ak + \gamma k^\constPower \right)^\constPower \\
& \leq \consteps(ak + \gamma k^\constPower) - \delta k^\constPower a^\constPower. 
\end{align*}
An identical argument applied to $G[R \cup S]$ yields a \twetamod  $Z_R$ of $G[R \cup S]$, such that
\begin{align*}
|Z_R| & \leq \consteps\left((1-a)k + \gamma k^\constPower \right) - \delta k^\constPower (1-a)^\constPower. 
\end{align*}

We now make a \twetamod  $Z$ of $G$ as follows. Let $Z = Z_L \cup S \cup Z_R$. The set $Z$ is a \twetamod  of $G$ because every connected component of $G- Z$ is a subset of $L$ or $R$, and $Z_L$ and $Z_R$ are \twetamod\!\!s for $G[L \cup S]$ and $G[R \cup S]$ respectively. Finally we bound the size of $Z$.
\begin{align*}
|Z| & \leq |Z_L| + |Z_R| + |S| \\
& \leq  \consteps(ak + \gamma k^\constPower) - \delta k^\constPower a^\constPower +  \consteps\left((1-a)k + \gamma k^\constPower \right) - \delta k^\constPower (1-a)^\constPower + \gamma k^\constPower \\
& = \consteps k - \delta k^\constPower \left((1-a)^\constPower + a^\constPower\right) + k^\constPower \gamma(2\consteps + 1) \\
& \leq \consteps k - \delta k^\constPower + k^\constPower \left(\gamma(2\consteps + 1) - \delta \rho \right) \\
& = \consteps k - \delta k^\constPower
\end{align*}
In the transition from the third to the fourth line we used that $(1-a)^\constPower + a^\constPower - 1 \geq \rho$ for any $a$ between $\frac{1}{3}$ and $\frac{2}{3}$. The claim follows.
\medskip

To conclude the proof, 
  we observe that the statement of the lemma follows from the above claim. If $OPT_\Pi(G) \leq k_0$ then $\tw(G)\leq \alpha \cdot k_0^\constPower=\eta$ and the empty set is a  \twetamod  of $G$ of size $0\leq  \consteps \cdot OPT_\Pi(G)$. If  $OPT_\Pi(G) > k_0$ then $G$ has a \twetamod  of size at most $\consteps \cdot OPT_\Pi(G) - \delta (OPT_\Pi(G))^\constPower\leq \consteps \cdot OPT_\Pi(G) $. This completes the proof.
\end{proof}

An identical argument yields an analogous lemma for contraction-bidimensional problems.
The only difference is that we now use Lemma~\ref{lem:bidimParameterTreewidthSep} instead of Lemma~\ref{lem:bidimParameterTreewidth}.
\begin{lemma}\label{lem:contrBidimTwModulator}
For any real $\consteps > 0$ and contraction-bidimensional linear-separable problem $\Pi$ on graph class ${\cal G}$ with the SQGC property, there exists an integer $\eta \geq 0$ such that any graph $G \in {\cal G}$ has a \twetamod  $S$ of size at most $\consteps \cdot OPT^*_\Pi(G).$
\end{lemma}

\subsection{Protrusion decomposition}\label{subsec:protr_decomps}
 The notions  of 
 \emph{protrusion}  and  \emph{protrusion decomposition} were introduced in \cite{BodlaenderFLPST16meta}.

\begin{definition}{\rm [\bf $t$-protrusion]}
For a graph $G$, a set $X \subseteq V(G)$ is a {\em $t$-protrusion} of $G$ if $|\partial(X)|\leq t$ and $\tw(G[X]) \leq t.$
\end{definition}

\begin{definition}{\rm [{\bf $(\alpha,r)$-Protrusion decomposition]}}
An $(\alpha,r)${\em -protrusion decomposition} of a graph $G$
is  a sequence ${\cal P}=\langle R_{0},R_{1},\ldots,R_{\ell}\rangle$ of pairwise disjoint subsets of $V(G)$ 
such that \begin{itemize}
\item $\bigcup_{i\in\{1,\ldots,\ell\}}=V(G),$
\item $\max\{\ell, |R_{0}|\}\leq \alpha,$  
\item each $R^{+}_{i}=N_{G}[R_i],$ $i\in\{1,\ldots,\ell\},$ is an $r$-protrusion of $G,$ and 
\item for every ${i\in\{1,\ldots,\ell\}},\ N_{G}(R_{i})\subseteq R_0.$ 
\end{itemize}
We call the sets $R^{+}_{i},$ $i\in\{1,\ldots,\ell\},$ the  {\em protrusions} of ${\cal P}$ and the set $R_0$ the {\em core} of ${\cal P}$.
\end{definition}

Next we prove that the existence of  a \twetamod  $S$ implies the existence of a protrusion decomposition of $G$ into $r$-protrusions with a core  not much bigger than $S$, and $r$ only depending on $\eta$ (and the graph class ${\cal G}$ we are working with). The inductive proof of the following lemma is similar to the proof of Lemma~\ref{lem:bidimTwModulator} about \twetamod, however this time we use  induction to construct the required  protrusion decomposition.  

\begin{lemma}\label{lem:protrus_decomp}
Let  ${\cal G}$ be a graph class with the SQGM property. For any real $\consteps > 0$ and positive integer $\eta$, there exists an integer $r$ such that if $G \in {\cal G}$ has a non-empty \twetamod  $S$, then $G$ has $((1+\consteps)|S|,r)$-protrusion decomposition ${\cal P}$ with $S$ contained in the core of ${\cal P}$.
\end{lemma}

\begin{proof}
Let $\alpha'>0$ and $\frac{1}{2} \leq \constPower < 1$ be the constants from Observation~\ref{obs:etaTransversalTreewidth}, in particular any graph in ${\cal G}$ with a non-empty  \twetamod  of size at most $k$ has treewidth at most $\alpha'  k^\constPower$. Set $\alpha = \max\{\alpha', 1\}$. 
We define a series of constants. As in the beginning of the 
proof of Lemma~\ref{lem:bidimTwModulator} the purpose of these constants will become  clear during the course of the proof. 

\begin{itemize}
\item We set $\rho = \frac{1^\constPower + 2^\constPower - 3^\constPower}{3^\constPower}$ and note that $\rho > 0$.
\item We define  $\delta = \frac{(1 + \consteps)4\alpha}{\rho}$,  and 
\item  $k_0 = (3+6\alpha)^\frac{1}{1-\constPower} + 3 \cdot (\frac{\delta}{\consteps})^\frac{1}{1-\constPower}$.
 
\item   Finally, we set $r = \max\{k_0, \alpha k_0^\constPower, \eta\}$
\end{itemize}
We first proof, using induction, the following claim:\medskip

\noindent{\em Claim:} {for every $k \geq \frac{1}{3}k_0$, if a  graph $G \in {\cal G}$  has a non-empty \twetamod  $S$ of size at most $k$, then $G$ has $((1+\consteps) k - \delta k^\constPower, r)$-protrusion decomposition with $S$ contained in the core of this protrusion decomposition.  }\smallskip

%
\noindent{\em Proof of claim:} 
In the base case we consider any $k$ such that $\frac{1}{3}k_0 \leq k \leq k_0$. By Observation~\ref{obs:etaTransversalTreewidth},  any graph that has a non-empty  \twetamod  $S$ of size at most $k $ has treewidth at most  $\alpha'  k^\constPower\leq \alpha  k_0^\constPower\leq r$. Consider the protrusion-decomposition ${\cal P}=\langle S,V(G)\setminus S\rangle$, i.e., ${\cal P}$ has core $S$ and only one protrusion, namely $V(G) \setminus S$. Since $k\leq k_0 \leq r$,   
 this is a $(k,r)$-protrusion decomposition. To complete the proof of the base case we need to show that $k \leq (1+\consteps)k - \delta k^\constPower$.  As in Lemma~\ref{lem:bidimTwModulator}, by the choice of  $k_0$, we have that  $0 \leq \consteps\frac{1}{3}k_0 - \delta \left(\frac{1}{3}k_0\right)^\constPower$ (see~\eqref{eq:k02}).  Hence 
\[
0 \leq \consteps\frac{1}{3}k_0 - \delta \left(\frac{1}{3}k_0\right)^\constPower \leq \consteps k - \delta k^\constPower\Rightarrow k \leq (1+\consteps)k - \delta k^\constPower.
\] 
 In the last inequality we used that for any $\frac{1}{2} \leq \constPower < 1$, $\consteps > 0$ and $\delta>0$ the function $\consteps k - \delta k^\constPower$ is monotonically increasing from the first point where it becomes positive. 
 Thus the base case follows.

For the inductive step we let $k > k_0$ and assume that the statement holds for all values less than $k$.  
 To prove the statement for $k$, let us  consider a graph $G \in {\cal G}$ and a non-empty \twetamod  $S$ of $G$ of size at most $k$. If $|S| < k$ then by the induction hypothesis, $G$ has a $((1+\consteps)(k-1) - \delta (k-1)^\constPower, r)$-protrusion decomposition. Since the function $(1+\consteps)k - \delta k^\constPower$ is non-decreasing whenever it is non-negative and $k-1 \geq k_0$ it follows that this is also a $((1+\consteps)k - \delta k^\constPower, r)$-protrusion decomposition. Thus we may assume that $|S| = k$.

By Observation~\ref{obs:etaTransversalTreewidth}, the treewidth of $G$ is at most $\tw(G) \leq \alpha k^\constPower$. We apply  Lemma~\ref{lemma:balsep22}  to $(G, S)$ and obtain  a 2/3-balanced separation $(A_1, A_2)$ of $(G, S)$ of order at most $\tw(G) + 1 \leq \alpha \cdot k^\constPower + 1$. Let $L = A_1 \setminus A_2$, $X = A_1 \cap A_2$ and $R = A_2 \setminus A_1$. 
As $|X|\leq \tw(G)+1$ we obtain that $|X|\leq \alpha\cdot  k^\constPower + 1\leq 2\alpha k^\constPower $ (here we use the fact that $k\geq \frac{k_0}{3}\geq 1$).
 Since $(A_1, A_2)$ is a 2/3-balanced separation,  it follows that there exists a real $\frac{1}{3} \leq a \leq \frac{2}{3}$ such that $|L \cap S| \leq a|S|$ and  $|R \cap S| \leq (1-a)|S|$.

Consider now the graph $G[L \cup X]$. Let $S_L = (S \cap L) \cup X$. Clearly $S_L$ is a \twetamod  of $G[L \cup X]$ and  $|S_L| =|S\cap L|+|X|\leq a|S| + |X| \leq ak + 2\alpha k^\constPower$.
In order to proceed with the induction, we have to verify that  for $k>k_0$,  
 
 \begin{equation}\label{eq:k03}
ak  + 2\alpha k ^\constPower \leq k  - 1 
  \end{equation}
  The proof of \eqref{eq:k03} is very similar to the proof of \eqref{eq:k01} from Lemma~\ref{lem:bidimTwModulator}. 
By the choice of $k_0$, we have 
\begin{eqnarray}
k_0 \geq  (3+6\alpha)^\frac{1}{1-\constPower}\Rightarrow k_0 \geq (\frac{3}{k_0^\constPower}+6\alpha)^\frac{1}{1-\constPower} \Rightarrow \nonumber\\k_0^{1-\constPower} \geq \frac{3}{k_0^\constPower}+6\alpha\Rightarrow
 \frac{k_0}{3}\geq 1+2\alpha k_{0}^{\constPower}\Rightarrow\nonumber\\ 
 \frac{2}{3}k_0 + 2\alpha k_0^\constPower \leq k_0 - 1\label{tyso}
\end{eqnarray} 
%
%
Now, because of~\eqref{tyso}, it is easy to   verify by differentiation that the inequality $\frac{2}{3}k  + 2\alpha k ^\constPower \leq k  - 1$ holds for every $k>k_0$. Thus \eqref{eq:k03} follows.  Notice also that $S_{L}$ is non-empty because 
$$|S_L| \geq |S \setminus R| =|S|-|S\cap R|\geq |S|-(1-a)|S|= a|S| \geq \frac{1}{3}k \geq \frac{1}{3}k_0 \geq 1,$$
and therefore  we may apply the induction hypothesis to $G[L \cup X]$  with $S_L$ as \twetamod. We  obtain a $((1+\consteps)|S_L|-\delta|S_L|^\constPower, r)$-protrusion decomposition ${\cal P}_L$ of $G[L \cup X]$ with core containing $S_L$.

We now consider the graph $G[R \cup X]$ and define  $S_R = (S \cap R) \cup X$. Working symmetrically to the case of $S_L$ it is possible to deduce  that $|S_R| \leq (1-a)k + 2\alpha k^\constPower$,  $(1-a)k  + 2\alpha k ^\constPower \leq k  - 1$, and $|S_R|\geq 1$, therefore we can apply the induction hypothesis to $G[R \cup X]$ and  have a  $((1+\consteps)|S_R|-\delta|S_R|^\constPower, r)$-protrusion decomposition  ${\cal P}_R$ of $G[R \cup X]$ with core containing $S_R$.

From the protrusion decompositions ${\cal P}_L$ and ${\cal P}_R$ we construct a protrusion decomposition ${\cal P}$ of $G$. The core of ${\cal P}$ is the union of the cores of ${\cal P}_L$ and ${\cal P}_R$. The set of protrusions of ${\cal P}$ is the union of the set of protrusions of ${\cal P}_L$ and ${\cal P}_R$ respectively. Since the cores of ${\cal P}_L$ and ${\cal P}_R$ contain $X$,  the protrusions of ${\cal P}_L$ and ${\cal P}_R$ are also protrusions in $G$. Therefore ${\cal P}$ is a $((1+\consteps)(|S_L|+|S_R|)-\delta(|S_L|^\constPower + |S_R|^\constPower), r)$-protrusion decomposition of $G$ containing $S$ in its core. Thus, to finish the proof of the claim  it is sufficient to show the following 
\begin{eqnarray}
(1+\consteps)(|S_L|+|S_R|)-\delta(|S_L|^\constPower + |S_R|^\constPower) \leq (1+\consteps)k - \delta k^\constPower.\label{op5o}
\end{eqnarray}
We now proceed with the proof of~\eqref{op5o}.
Since $|S_L| \leq ak + 2\alpha k^\constPower$ and $|S_R| \leq(1-a)k + 2\alpha k^\constPower$ it follows that 
\begin{eqnarray}
(1+\consteps)(|(S_L|+|S_R|) \leq (1+\consteps)(k + 4\alpha k^\constPower).\label{ehop}
\end{eqnarray}
Recall that $S_{L}=S\setminus (R\cap S)$, therefore $|S_{L}|=|S|-|R\cap S|\geq |S|-(1-a)|S|= a|S|=ak$. Working analogusly on $S_{R}$, we can show that
$|S_{R}|\geq (1-a)k$. The two last inequalities imply that  $|S_{L}|^{\constPower}+|S_{R}|^{\constPower}\geq (ak)^{\constPower}+((1-a)k)^{\constPower}$, therefore 
\vspace{-8mm}

\begin{eqnarray}
\delta(|S_L|^\constPower + |S_R|^\constPower) & \geq & \delta k^\constPower (a^\constPower + (1-a)^\constPower).\label{opio}
\end{eqnarray}
Departing from~\eqref{ehop} and~\eqref{opio} we prove~\ref{op5o} as follows.
\begin{align*}
& ~(1 + \consteps)(|S_L|+|S_R|)-\delta(|S_L|^\constPower + |S_R|^\constPower) \\
\leq & ~ (1+\consteps)(k + 4\alpha k^\constPower) -  \delta k^\constPower (a^\constPower + (1-a)^\constPower) \\
= & ~ (1+\consteps)k - \delta k^\constPower + k^\constPower(4\alpha(1+\consteps) - \delta(a^\constPower + (1-a)^\constPower - 1)) \\
\leq & ~ (1+\consteps)k - \delta k^\constPower + k^\constPower(4\alpha(1+\consteps) - \rho\delta) \\
= & ~ (1+\consteps)k - \delta k^\constPower
\end{align*}
In the transition from the third to the fourth line we used that $(1-a)^\constPower + a^\constPower - 1 \geq \rho$ for any $a$ between $\frac{1}{3}$ and $\frac{2}{3}$. This completes the proof of the claim.
\medskip

We have now proved that for any $k \geq \frac{1}{3}k_0$ and for every graph $G \in {\cal G}$,   
if 
$G$ has a non-empty  \twetamod  $S$ of size at most $k$, then $G$ has  an  
$((1+\consteps) k - \delta k^\constPower, r)$-protrusion decomposition, which is also a  $((1+\consteps)k, r)$-protrusion decomposition containing $S$ in its core.

In the remaining case, where $G$ has a non-empty \twetamod  $S$ and $|S| \leq \frac{k_0}{3}$ then the protrusion decomposition ${\cal P}=\langle S,V(G) \setminus S\rangle$, where $S$ is the core and $V(G) \setminus S$ is the unique protrusion, is an $(|S|, \eta)$-protrusion decomposition. Since $|S| \leq (1+\consteps)|S|$ and $\eta \leq r$ this completes the proof of the lemma.
\end{proof}

Lemmata~\ref{lem:bidimTwModulator},~\ref{lem:contrBidimTwModulator}, and~\ref{lem:protrus_decomp}
along with Observation~\ref{ovsSQGCSQGM},  imply the following theorem, which is the first main technical contribution of this paper. 

%
%

%

\begin{theorem}\label{thm:protrusiondecomp}
Let  ${\cal G}$ be a graph class with the SQGM (resp. SQGC) property and $\Pi$ be a minor-bidimensional (resp. contraction-bidimensional) linear-separable problem. 
There exists a constant $c$ such that every graph 
  $G\in {\cal G}$,
  admits a $(c\cdot k,c)$-protrusion decomposition, where $k=OPT_{\Pi\doublecap{\cal G}}(G)$.
\end{theorem}
\begin{proof}
Let $\Pi$ be a minor-bidimensional linear-separable problem and $k=OPT_{\Pi\doublecap{\cal G}}(G)$. 
We apply Lemma~\ref{lem:bidimTwModulator} for $\consteps=1$, 
and we deduce that   there exists an integer $\eta \geq 0$ such that   graph $G $ has a \twetamod  $S$ of size at most $k.$ 
If $|S|> 0$, then we can apply Lemma~\ref{lem:protrus_decomp}, for $\consteps=1$
and obtain that  there exists an integer $r$ such that $G$ admits a $(2\cdot k,r)$-protrusion decomposition. 
If $S=\emptyset$, then $\tw(G)\leq \eta$ and $\langle S,V(G)\rangle$ is trivially an $(2\cdot k, \eta)$-protrusion decomposition of $G$.
In any case, by setting $c=\max\{2,r,\eta\}$, we have that $G$ admits a $(c\cdot k,c)$-protrusion decomposition, as required.

The proof for the case where $\Pi$ is contraction-bidimensional 
is the same as above with the difference that we now apply Lemma~\ref{lem:contrBidimTwModulator}
instead of Lemma~\ref{lem:bidimTwModulator}  and, before we apply Lemma~\ref{lem:protrus_decomp},
we use Observation~\ref{ovsSQGCSQGM} in order to show that ${\cal G}$ has the SQGM property.
\end{proof}
%
%

%
\section{Finite index and finite integer index}\label{sec:FII}

In this section we prove the second main technical contribution of the paper: ``CMSO $+$ separability $\Rightarrow$ finite integer index''. 
\subsection{Definitions on boundaried graphs} 
\paragraph{Boundaried graphs.}
A {\em boundaried graph} ${\bf G}=(G,B,\lambda)$ is a triple consisting of a graph $G$, a set $B\subseteq V(G)$ 
of  distinguished vertices and an injective labelling ${\lambda}$
from $B$  to the set $\Bbb{Z}_{> 0}$. The set $B$ is called the {\em {boundary}} of ${\bf G}$ and  the vertices in $B$  are called  {\em boundary vertices} or {\em {terminals}}. Also, the graph $G$ is the {\em underlying graph} of ${\bf G}$. 
Given a boundaried graph ${\bf G}=(G,B,\lambda)$,
 we define its {\em label set} by ${\Lambda({\bf G})}=\{\lambda(v)\mid v\in B\}$.
If ${\bf G}=(G,B,\lambda)$  is a boundaried graph and $S\subseteq V(G)$, then  the pair
$\hat{{\bf G}}=({\bf G},S)$ is an {\em annotated boundaried graph} and we say that $G$ is the {\em underlying graph} of $\hat{\bf G}$,
$B$ is its {\em boundary} (labeled by $\lambda$), and $S$ is its {\em annotated set}.

Given a finite set $I\subseteq \Bbb{Z}_{> 0}$, we define 
${{\cal F}_{I}}$ (resp. $\hat{{\cal F}_{I}}$)  as  the class of all boundaried graphs (resp. annotated boundaried graphs) whose label set is $I$.
Similarly, we define ${\cal F}_{\subseteq I}=\bigcup_{I'\subseteq I}{\cal F}_{I'}$ (resp. $\hat{\cal F}_{\subseteq I}=\bigcup_{I'\subseteq I}\hat{\cal F}_{I'}$).
We also denote by ${{\cal F}}$ (resp. $\hat{\cal F}$) the class of all boundaried graphs, i.e., ${\cal F}={\cal F}_{\subseteq \Bbb{Z}_{> 0}}$ (resp.  $\hat{\cal F}=\hat{\cal F}_{\subseteq \Bbb{Z}_{> 0}}$).
Finally we say that a boundaried graph ${\bf G}$ is a {\em $t$-boundaried} graph if $\Lambda({\bf G})\subseteq \{1,\ldots,t\}$.

Let ${\cal G}$ be a class of (not boundaried)  graphs.
We say that a boundaried graph ${\bf G}$ {\em belongs to  ${\cal G}$} if the underlying graph of ${\bf G}$ belongs to ${\cal G}.$
We also use $V({\bf G})$ to denote the vertex set of the underlying graph of ${\bf G}$.

%

%
%

\paragraph{The gluing operation.}
Let ${\bf G}_1=(G_{1},B_{1},\lambda_{1})$ and ${\bf G}_2=(G_{2},B_{2},\lambda_{2})$ be two  boundaried graphs. We denote by ${\bf G}_1 {\oplus} {\bf G}_2$ the  graph 
(not boundaried) obtained by taking the disjoint union of $G_1$ and $G_2$ and identifying equally-labeled vertices
of the boundaries of $G_{1}$ and $G_{2}.$ The gluing operation maintains edges of both graphs that are glued, i.e.,
in ${\bf G}_1 \oplus {\bf G}_2$ there is an edge between two labeled 
vertices if there is either an edge between them in $G_1$ or in $G_2,$  (in case of multi-graphs,
multiplicities of edges are summed up in the new graph).  
%

Let $G={\bf G}_{1}\oplus {\bf G}_{2}$ where ${\bf G}_1=(G_{1},B_{1},\lambda_{1})$ and ${\bf G}_2=(G_{2},B_{2},\lambda_{2})$ are boundaried graphs.
We define the {\em glued} set of ${\bf G}_{i}$ as the set $B_{i}^\star=\lambda_{i}^{-1}(\Lambda({\bf G}_{1})\cap \Lambda({\bf G}_{2})), i=1,2$. For a vertex $v\in V({\bf G}_{1})$ we define its {\em {\em heir}} ${{\sf heir}(v)}$ in 
$G$ as follows: if $v\not\in B_{1}^\star$ then ${\sf heir}(v)=v$, otherwise ${\sf heir}(v)$ is the result of the identification 
of $v$ with an equally labeled vertex in ${\bf G}_{2}$. The {\em heir} of a vertex in $G_{2}$ is defined symmetrically. The {\em {\em common boundary}} of ${\bf G}_{1}$ and ${\bf G}_{2}$ in $G$ is equal 
to ${\sf heir}(B_{1}^\star)={\sf heir}(B_{2}^\star)$ where the evaluation of {\sf heir} on vertex sets is defined in the obvious way.

Let now $\hat{\bf G}_{1}=({\bf G}_{1},S_{1})$ and 
 $\hat{\bf G}_{2}=({\bf G}_{2},S_{2})$ be two annotated boundaried graphs.
 We define $\hat{\bf G}_{1}\oplus \hat{\bf G}_{2}$ as the annotated graph $\hat{G}=({\bf G}_1 \oplus {\bf G}_2,{\sf heir}(S_{1})\cup {\sf heir}(S_{2}))$.  


\subsection{Finite Index}
Let $\phi$ be a predicate on annotated graphs. In other words $\phi$ is a function that takes as input a graph $G$ and vertex set $S \subseteq V(G)$, and outputs true or false. We define a {{\em canonical equivalence relation}} ${\equiv_\phi}$ on boundaried annotated graphs as follows. For two annotated boundaried graphs $\hat{\bf G}_{1}=({\bf G}_{1},S_{1})$ and $\hat{\bf G}_{2}=({\bf G}_{2},S_{2})$, we say that $\hat{\bf G}_{1} \equiv_\phi \hat{\bf G}_{2}$ if  
$\Lambda({\bf G}_{1}) =  \Lambda({\bf G}_{2})$
and for every annotated boundaried graph $\hat{\bf G}=({\bf G},S)$ we have that
\begin{eqnarray*} 
\phi(\hat{\bf G}_{1}\oplus\hat{\bf G}) = {\sf true}  \Leftrightarrow  \phi(\hat{\bf G}_{2}\oplus\hat{\bf G}) = {\sf true} 
\end{eqnarray*}
It is easy to verify that $\equiv_\phi$ is an equivalence relation. We say that $\phi$ is {{\em finite state}} if, for 
every  finite $I\subseteq \Bbb{Z}_{>0},$ the equivalence relation $\equiv_\phi$ has a finite number of equivalence classes when 
restricted to $\hat{\cal F}_{I}$.  
%
A formal proof of the following fact can be found in~\cite{BodlaenderFLPST16meta,Courcelle97}.
\begin{proposition}\label{propCMSOann}
For every CMSO-definable predicate $\phi$ on annotated graphs, $\phi$ has finite state.
\end{proposition}
Given a  CMSO-definable predicate $\phi$ on annotated graphs and an $I\subseteq \Bbb{Z}_{>0}$, we use the notation $\hat{\cal R}_{\phi,I}$ for a set containing one minimum-size representative from each of the equivalence classes of $\equiv_{\phi}$ when restricted to $\hat{\cal F}_{I}$. 

\subsection{Finite Integer Index} 
\label{subsec:finiinteginde}

\begin{definition}{\rm [\bf Canonical equivalence on boundaried graphs.]}
Let $\Pi$ be a parameterized graph problem whose instances are pairs of the form $(G,k).$
 Given two boundaried graphs ${\bf G}_1,{\bf G}_2~\in {\cal F},$ we say that \term{${\bf G}_1\!\equiv _{\Pi}\! {\bf G}_2$} if 
$\Lambda({\bf G}_{1})=\Lambda({\bf G}_{2})$
 and there exists a \term{{\em transposition constant}}
$c\in\Bbb{Z}$ such that 
\begin{eqnarray*}
\forall({\bf F},k)\in {\cal F}\times \Bbb{Z} &&  ({\bf G}_1 \oplus {\bf F}, k) \in \Pi \Leftrightarrow ({\bf G}_2 \oplus {\bf F}, k+c) \in \Pi.\label{eq:fiidef}
\end{eqnarray*}
\end{definition}
Note that  the relation $\equiv_{\Pi}$  is
an equivalence relation. Observe that $c$ could be negative in the above definition. This is the reason we gave the definition of parameterized problems to include negative parameters also.

Notice that two  boundaried graphs with different label sets belong to 
different equivalence classes of $\equiv_{\Pi}.$ 
We are now in position  to give the following definition.

\begin{definition}{\rm [\bf Finite Integer Index]}
\label{def:deffii}
A parameterized graph problem $\Pi$ whose instances are pairs of the form $(G,k)$
has {\em Finite Integer Index} (or simply has \term{{\em FII}}), if and only if for every finite $I\subseteq \Bbb{Z}^+,$
the number of equivalence classes of  $\equiv_{\Pi}$ that are subsets of ${\cal F}_{I}$
is finite. 
\end{definition} 

The notion of FII first appeared in the works of ~\cite{BodlaendervA01a,Fluiter97} and is similar to the
notion of {\em finite state} \cite{AbrahamsonF93,BoriePT92,Courcelle90}. 

%
%


\subsection{A condition for proving  FII}
\begin{theorem}
\label{fiiopoiok}
If $\Pi$ is a separable CMSO-minimization or CMSO-maximization problem, then $\Pi$ has FII.
\end{theorem}

\begin{proof} 
We prove the statement for vertex subset CMSO-minimization problems. The proofs for CMSO-maximization problems and the edge subset variants are identical. Let $\consteps$ be the CMSO-predicate defining $\Pi$, that is $(G,k) \in \Pi$ if and only if there exists a vertex set $S$ of size at most $k$ such that $(G,S) \models \consteps$. Let  $I$ be a finite subset of $\Bbb{Z}_{>0}$ and let $\hat{\cal R}_{\consteps,I}$ be a
 set of minimum-size representatives of $\equiv_{\consteps}$ when restricted to $\hat{\cal F}_{I}$, the set of annotated boundaried graphs with label set $I$. 
By Proposition~\ref{propCMSOann},  the set $\hat{\cal R}_{\consteps,I}$  is finite. 
Let $f$ be a function such that $\Pi$ is $f$-separable.
 
Given an annotated  boundaried graph ${\bf G}\in \hat{\cal F}_{I}$, we define the function $\xi_{\bf G}: \hat{\cal R}_{{\consteps},I}\rightarrow \Bbb{Z}_{\geq 0} \cup \{\bot\}$ such that
if  $({\bf G}^*,S^*)\in\hat{\cal R}_{{\cal \consteps},I}$, then
\begin{eqnarray}
\xi_{\bf G}({\bf G}^*,S^*)&=&  {\rm min}\{|S|\mid ({\bf G},S)\equiv_{\consteps}({\bf G}^*,S^*)\}. \label{fpes}
\end{eqnarray}
Here $\xi_{\bf G}({\bf G}^*,S^*) = \bot$ if no set $S$ such that $({\bf G},S)\equiv_{\consteps}({\bf G}^*,S^*)$ exists. We may think of $\xi_{\bf G}$ as a partial function where $\bot$ means that the function is left undefined. We also define   function $\chi_{\bf G}({\bf G}^*,S^*) : \hat{\cal R}_{{\consteps},I} \rightarrow \Bbb{Z} \cup \{\bot\}$ as follows
\begin{eqnarray*}
\chi_{\bf G}({\bf G}^*,S^*)= 
\begin{cases}
\xi_{\bf G}({\bf G}^*,S^*)-OPT_\Pi(G),& \text{if }  \xi_{\bf G}({\bf G}^*,S^*)\in [OPT_\Pi(G)-f(|I|), OPT_\Pi(G)+f(|I|)] \\
    \bot,              & \text{otherwise}.
\end{cases}
%
%
%
\end{eqnarray*}
Thus  $\chi_{\bf G}$ outputs $\bot$ if $\xi_{\bf G} = \bot$ or  when 
$ \xi_{\bf G}({\bf G}^*,S^*)\not\in [OPT_\Pi(G)-f(|I|), OPT_\Pi(G)+f(|I|)]$.
Clearly $\chi_{\bf G}({\bf G}^*,S^*) \in [-f(|I|), f(|I|)] \cup \{\bot\}$ for all choices of $({\bf G}^*,S^*) \in \hat{\cal R}_{{\consteps},I}$. This means that $\chi_{\bf G}$ may take at most $2\cdot f(|I|)+2$ different values, and hence there are at most $(2\cdot f(|I|)+2)^{|\hat{\cal R}_{{\consteps},I}|}$ different functions $\chi_{{\bf G}}$.

Now, we define an equivalence relation $\sim$ on ${\cal F}_{I}$ such that, 
given ${\bf G}_{1},{\bf G}_{2}\in {\cal F}_{I}$, 
$${\bf G}_{1}\sim {\bf G}_{2}\text{~if and only if ~}\chi_{{\bf G}_{1}}=\chi_{{\bf G}_{2}}$$
and keep in mind that $\sim$ has a  finite number of equivalence 
classes.
We will prove that $\sim$  is a refinement  of $ \equiv _{\Pi}.$
For this we prove that if ${\bf G}_{1}\sim {\bf G}_{2},$ then ${\bf G}_{1} \equiv _{\Pi} {\bf G}_{2}.$\medskip

We assume that ${\bf G}_{1}\sim {\bf G}_{2},$ where ${\bf G}_{1}=(G_{1},B_{1})$ and ${\bf G}_{2}=(G_{2},B_{2})$.
We also set  $c=OPT_\Pi(G_{2})-OPT_\Pi(G_{1})$ (notice that $c$ depends only  on ${\bf G}_{1}$ and ${\bf G}_{2}$). 
In what follows, 
we  prove that,  for every
$({\bf F},k)\in {\cal F}_{I}\times  \Bbb{Z}$,  $({\bf G}_{1}\oplus {\bf F},k)\in \Pi $   implies that $({\bf G}_{2}\oplus {\bf F},k+c)\in\Pi $ (the direction $({\bf G}_{1}\oplus {\bf F},k)\not \in \Pi \Rightarrow ({\bf G}_{2}\oplus {\bf F},k+c)\not\in\Pi $ is symmetric).

We consider some $({\bf F},k)\in {\cal F}_{I}\times  \Bbb{Z}$ and we denote $H={\bf G}_{1}\oplus {\bf F}$.
Let $S$ be an optimal solution of $H$. That is, $|S|= OPT_\Pi(H)$ and $(H,S)\models \consteps$. 
The fact that $(H,k)\in \Pi $ means that the size of $S$ is at most $k$. 
We define $S_{1}=S\cap V({\bf G}_{1})$ and $S^{\star}=S\setminus S_{1}$
and we know that $({\bf G}_{1}\oplus {\bf F},S_{1}\cup S^{\star})\models \consteps.$
Notice that $|S_{1}|+|S^{\star}|=|S_{1}\cup S^\star| \leq k$.\medskip

\noindent{\em Claim:} 
There exists no $S_{1}'\subseteq V({\bf G}_{1})$ such that $|S_{1}'| <|S_{1}|$ and $({\bf G}_{1},S_{1}')\equiv_{\consteps} ({\bf G}_{1},S_{1}).$ 

\noindent{\em Proof of Claim:} 
Suppose in contrary that such an $S_1'$ exists.
Since $({\bf G}_{1},S_{1}')\equiv_{\consteps} ({\bf G}_{1},S_{1}),$  we have  that 
$({\bf G}_{1}\oplus {\bf F},S_{1}'\cup S^{\star})\models \consteps.$ However, this implies that 
$|S_{1}'\cup S^{\star}|= |S_{1}'|+|S^{\star}| < |S_{1}|+|S^{\star}|=OPT_\Pi(H)$, which is  a contradiction. This concludes the claim. \qed

%
%

Let $({\bf G}',S')\in\hat{\cal R}_{\consteps,I}$ such that $({\bf G}',S')\equiv_{\consteps}({\bf G}_{1},S_{1}).$ 
By the 
above claim and~\eqref{fpes}, we have that $\xi_{{\bf G}_{1}}({\bf G}',S')=|S_1|.$  Let $A_{1}={\sf heir}(V({\bf G}_{1}))$ and $A_{2}={\sf heir}(V({\bf F}))$ (here by the ${\sf heir}$ of a vertex set we mean the union of all vertex ${\sf heir}$ it contains). Here, function ${\sf heir}$
is defined with respect to the operation $H={\bf G}_{1}\oplus {\bf F}$ and certainly $(A_{1},A_{2})$ is a separation of $H$
of order $|I|$.

Recall that $S_{1}=S\cap A_{1}$ and $|S|= OPT_\Pi(H)$.  
The $f$-separability of $\Pi$ implies that
\begin{eqnarray*}
 |S\cap A_{1}|  \in   [OPT_\Pi(H[A_{1}])-f(|I|),OPT_\Pi(H[A_{1}])+f(|I|)], 
\end{eqnarray*}
and therefore, 
\begin{eqnarray*}
   \xi_{{\bf G}_{1}}({\bf G}',S')   \in   [OPT_\Pi(G_{1})-f(|I|),OPT_\Pi(G_{1})+f(|I|)].
\end{eqnarray*}

Then  $\chi_{{\bf G}_{1}}({\bf G}',S') \neq \bot$ and thus $\chi_{{\bf G}_{1}}({\bf G}',S')=|S_{1}|-OPT_\Pi(G_{1})$. As ${\bf G}_{1}\sim {\bf G}_{2},$ this means that $\chi_{{\bf G}_{2}}({\bf G}',S')=|S_{1}|-OPT_\Pi({\bf G}_{1})$. Since $\chi_{{\bf G}_2}({\bf G}',S') \neq \bot$ it follows that $\chi_{{\bf G}_2}({\bf G}',S') =\xi_{{\bf G}_{2}}({\bf G}',S')-OPT_\Pi(G_{2})$. We conclude that
\begin{eqnarray}
\xi_{{\bf G}_{2}}({\bf G}',S') =  |S_1|+c. \label{j4kls}
\end{eqnarray}
By the definition of $\xi$  and ~\eqref{j4kls},  there exists a set $S_2$ of size $|S_{1}|+c$ such that $({\bf G}_{2},S_{2})\equiv_{\consteps} ({\bf G}',S')$.
As $({\bf G}_{1},S_{1})\equiv_{\consteps}({\bf G}_{2},S_{2}),$ the fact that $({\bf G}_{1}\oplus {\bf F},S_{1}\cup S^{\star})\models \consteps,$
implies that $({\bf G}_{2}\oplus {\bf F},S_{2}\cup S^{\star})\models \consteps.$ 
But then $|S_{2}\cup S^\star|= |S_{2}|+|S^{\star}| = |S_{1}|+c+|S^{\star}|\leq k+c$.
This means that $({\bf G}_{2}\oplus {\bf F},k+c)\in \Pi$. Hence  ${\bf G}_{1} \equiv _{\Pi} {\bf G}_{2}.$ 
\medskip

Thus for every pair of boundaried graphs ${\bf G}_{1}$ and ${\bf G}_{2},$ condition ${\bf G}_{1}\sim {\bf G}_{2},$ yields that  ${\bf G}_{1} \equiv _{\Pi} {\bf G}_{2}.$ Since $\sim$ has a finite number of equivalence classes, we conclude that $\Pi$ has FII.
\end{proof}

We remark that the proof of   Theorem~\ref{fiiopoiok} is similar in spirit to the proof of \cite[Lemma~7.3]{BodlaenderFLPST16meta} with a few key differences.  Being separable is a looser constraint for  CMSO-optimization problems than being strongly monotone (see \cite{BodlaenderFLPST16meta} for the definition.) In particular, separability only puts restrictions on how an \emph{optimum} solution can interact with both sides of the separation. On the other hand, strong  monotonicity  puts constraints on how \emph{any} solution can interact with the two sides.  Therefore, it is often easier to verify that a CMSO-optimization problem is separable than that  it is strongly monotone.
However, strong monotonicity allows us to conclude even stronger properties than FII as observed in 
\cite{FominLMS12}.

\section{Proof of the main theorem: Putting things together}\label{sec:puttin}
Now everything prepared to pipeline the results about protrusion decomposition and FII  with the framework from  \cite{BodlaenderFLPST16meta}. We start from the following definitions. 

\begin{definition}{\rm [\bf $(f,a)$-protrusion replacement family]}
Let $\Pi$ be a parameterized graph problem, let $f:\Bbb{Z}^{+}\rightarrow \Bbb{Z}^{+}$ be a non-decreasing function and let $a\in\Bbb{Z}^{+}.$
An {\em $(f,a)$-protrusion replacement family}  for $\Pi$ is a collection ${\cal A}=\{{\sf A}_{i}\mid i\geq 0\}$ of algorithms, 
such that  algorithm ${\sf A}_{i}$ receives as input a pair $(I,X),$ where 
\begin{itemize}
\item  $I$ is an instance  of $\Pi$ whose graph and parameter are $G$ and $k\in\Bbb{Z},$  
\item $X$ is an $i$-protrusion  of $G$ with at least $f(i)\cdot k^{a}$
vertices, 
\end{itemize}
and outputs
an equivalent instance $I^{*}$
such that, if $G^*$ and $k^*$  are the graph and the parameter 
of $I^*,$  then  $|V(G^{*})|<|V(G)|$ and $k^{*}\leq k.$
\end{definition}



The following  two properties for a  
parameterized graph problem $\Pi$ were defined  in \cite{BodlaenderFLPST16meta}. 
\begin{itemize}
\item[{\bf A}] {\rm [}{\bf Protrusion replacement}:{\rm ]} There exists an  $(f,a)$-protrusion replacement family ${\cal A}$ for $\Pi,$ for some function $f: \Bbb{Z}^{+}\rightarrow \Bbb{Z}^{+}$
and some $a\in\Bbb{Z}^{+}.$
\item[{\bf B}] {\rm [}{\bf Protrusion decomposition:}{\rm ]} There exists a constant $c$ such that, 
if  $G$ and $k\in\Bbb{Z}^+$ are the graph and the parameter of a \yesinstance of $\Pi$
then $G$ admits a $(c\cdot OPT_\Pi(G),c)$-protrusion decomposition.
\end{itemize}

Our kernelization result  is based on the following master theorem from \cite{BodlaenderFLPST16meta}.

\begin{theorem}
\label{master1}
If a parameterized graph problem $\Pi$ has property {\bf A} for some nonnegative constant $a$ 
and property {\bf B} for some constant $c,$ then $\Pi$ admits a kernel of size $O(k^{a+1}).$ 
\end{theorem}

The following lemma proven in   \cite{BodlaenderFLPST16meta}, shows that every parameterized problem that has  FII admits 
$(f,0)$-protrusion replacement families.
\begin{lemma}
\label{lem:fiiwithaequal0}
Every parameterized graph problem $\Pi$ that  has {FII} has the  protrusion  replacement  property  {\bf A}  for $a=0.$
\end{lemma}

\begin{proof}[Proof of Theorem~\ref{thm:main_result_bidim}]
Let $\Pi$ be a  CMSO-definable linear-separable minor-bidimensional problem on an $H$-minor-free  graph class ${\cal G}$ for some fixed graph $H$. By Proposition~\ref{prop:linear_grid_minor}, 
${\cal G}$  has the SQGM property.

By Theorem~\ref{fiiopoiok}, every 
 separable CMSO-minimization or CMSO-maximization problem   $\Pi$   has FII. Thus by Lemma~\ref{lem:fiiwithaequal0}, $\Pi$   has the  protrusion  replacement  property  {\bf A}  for $a=0.$
By Theorem~\ref{thm:protrusiondecomp},  $\Pi$ has the  protrusion decomposition property  {\bf B}.
  By Theorem~\ref{master1}, $\Pi$ admits a linear kernel.

 The proof for contraction-bidimensional problems is almost identical.
\end{proof}

\section{Conclusion}\label{sec:conclusion}
We conclude with a discussion on the relation between  the new and previous kernelization meta-theorems, the running times of kernelization algorithms and an open question.

 \medskip\noindent\textbf{Bidimensionality vs. 
quasi-coverability.}
Let us finally note how the techniques developed to prove the kernelization meta-theorem 
in this paper, can be used to refine  the results from   \cite{BodlaenderFLPST16meta}.

Let $r$ be a non-negative integer. 
We say that a parameterized graph problem $\Pi$
has the  {\emph{radial $r$-coverability property} }  if all  
\yesinstances of $\Pi$ encode graphs embeddable 
in some surface of Euler genus at most $r$ and there exist such an embedding of $G$
and  a
set $S \subseteq V(G)$ such that $|S|\leq r \cdot k$ and  ${\bf R}_{G}^{r}(S)=V(G)$ (here we denote by ${\bf R}_{G}^{r}(S)$ the set of all vertices of $G$ that are within radial distance\footnote{The {\em radial distance} 
between two vertices $x,y$ in a surface-embedded graph is one less than the minimum size of a sequence 
of alternating vertices and faces, starting on $x$ and finishing on $y$, where if an vertex $v$ and a face $f$ appear consecutively in the sequence, then $v$ is a vertex of the boundary of $f$.}  at most $r$ from some vertex in $S$).
 A parameterized graph problem $\Pi$ has the {\emph{radial $r$-quasi-coverability property}}  
if  all  
\yesinstances of $\Pi$ encode graphs embeddable 
in some surface of Euler genus at most $r$ and there exist such an embedding 
and  a
set $S \subseteq V(G)$ such that $|S|\leq r \cdot k$ and $\tw(G- {\bf R}_{G}^{r}(S))\leq r$.
 Every problem $\Pi$ that has the radial $r$-coverability property is radially $r$-quasi-covervable. The converse is not necessarily true. 
 

The main two meta-theorems from  \cite{BodlaenderFLPST16meta} say that 
quasi-coverable problems with finite   integer index admit   linear kernels  on graphs of bounded genus
and that coverable CMSO-definable problems admit  polynomial kernels on graphs of bounded genus.
By making use of SQGC property of graphs of bounded genus, it is not difficult to prove  that if a set $S$ is such that  $\tw(G- {\bf R}_{G}^{r}(S))\leq r$, then   $S$ is also a \twetamod for $\eta =\cO(r)$. The following lemma allows to  refine the results from 
 \cite{BodlaenderFLPST16meta}.
\begin{lemma}
If $\Pi$ is a quasi-coverable problem on graphs of bounded genus, then there exists a contraction-bidimensional separable problem $\Pi^*$ such that if $(G,k)\in \Pi$, then $OPT_{\Pi^*}(G)=\cO(k)$. 
\end{lemma}
\begin{proof}
The lemma follows from the following observation. Let $\Pi$ be an $r$-quasi-coverable problem. 
Consider the following problem $\Pi^*$:  for a  graph $G$ of genus $g$, pair $(G,k)\in \Pi^*$ if and only if there is a subset of vertices $S\subseteq V(G)$  such that  $|S|\leq r \cdot k$ and  $\tw(G- {\bf R}_{G}^{r}(S))\leq r$. 
 Because $\Pi$ is  $r$-quasi-coverable, we have that if $(G,k)\in \Pi$, then $OPT_{\Pi^*}(G)\leq  k$.
As in the case with \textsc{Treewidth-$\eta$-Modulator},  it is easy to see that $\Pi^*$ is contraction-bidimensional and separable. 
\end{proof}

By   Lemma~\ref{lem:protrus_decomp}, there is a protrusion decomposition for $\Pi^*$ and thus there is a protrusion decomposition  for $\Pi$.  
This implies one of the main 
result of  \cite{BodlaenderFLPST16meta}: Every  quasi-coverable problem with finite   integer index admits a linear kernel on graphs of bounded genus. This also can be used to show that every  quasi-coverable CMSO-definable problem admits a polynomial kernel on graphs of bounded genus. Thus the results of this paper subsume all the linear and polynomial kernels on graphs of bounded genus from \cite{BodlaenderFLPST16meta}. We refer to \cite{BodlaenderFLPST16meta}
for the list of these problems. 

  \medskip\noindent\textbf{Running time.} In Theorem~\ref{thm:main_result_bidim} we do not specify the running time of our kernelization algorithms. 
 The running time of such   algorithms depends on the running time hidden in the  lemma from \cite{BodlaenderFLPST16meta} (Lemma~\ref{master1}), which in turn, depends on how fast one can identify and replace protrusions in a graph. By 
applying   the fast ``protrusion replacer" from \cite{fomin2015solving},  see also \cite[Chapter~16]{kernelizationbook19}, it is possible to achieve kernelization algorithms in Theorem~\ref{thm:main_result_bidim} which run in time 
 linear  in the input size. 

\medskip

 \medskip\noindent\textbf{Open question.}
We conclude with the following open question. For separable contraction-bidimensional CMSO-definable problems, the technique developed in this paper yields the existence of a linear  kernel on $H$-minor-free graphs only when $H$ is an apex graph. An interesting  open question is to identify  general logic and combinatorial conditions for contraction-bidimensional problems which yield  a polynomial  kernelization on $H$-minor-free graphs. For some contraction-bidimensional problems, like \textsc{Dominating Set} or \textsc{Connected Dominating Set}, linear kernels are known to exist  
for
$H$-minor-free graphs and  
  $H$-topological-minor-free graphs~\cite{FominLST18kern}, and even for graphs of bounded expansion \cite{Drange16}, see also 
   \cite{EickmeyerGKKPRS17neig} but we are still far from a general understanding of the complexity of kernelization for such problems beyond apex-minor-free graphs.

\bibliographystyle{siam}


\begin{thebibliography}{10}

\bibitem{AbrahamsonF93}
{\sc K.~R. Abrahamson and M.~R. Fellows}, {\em Finite automata, bounded
  treewidth and well-quasiordering}, in AMS Summer Workshop on Graph Minors,
  Graph Structure Theory, Contemporary Mathematics vol. 147, N.~Robertson and
  P.~D. Seymour, eds., American Mathematical Society, 1993, pp.~539--564.

\bibitem{AFN04}
{\sc J.~Alber, M.~R. Fellows, and R.~Niedermeier}, {\em Polynomial-time data
  reduction for dominating set}, J. ACM, 51 (2004), pp.~363--384.

\bibitem{AlberFN04}
{\sc J.~Alber, H.~Fernau, and R.~Niedermeier}, {\em Parameterized complexity:
  exponential speed-up for planar graph problems}, J. Algorithms, 52 (2004),
  pp.~26--56.

\bibitem{AST90}
{\sc N.~Alon, P.~D. Seymour, and R.~Thomas}, {\em {A Separator Theorem for
  Nonplanar Graphs}}, Journal of the American Mathematical Society, 3 (1990),
  pp.~801--808.

\bibitem{AlonST94}
{\sc N.~Alon, P.~D. Seymour, and R.~Thomas}, {\em Planar separators}, SIAM J.
  Discrete Math., 7 (1994), pp.~184--193.

\bibitem{ArnborgLS91}
{\sc S.~Arnborg, J.~Lagergren, and D.~Seese}, {\em Easy problems for
  tree-decomposable graphs}, J. Algorithms, 12 (1991), pp.~308--340.

\bibitem{BasteT17cont}
{\sc J.~Baste and D.~M. Thilikos}, {\em Contraction-bidimensionality of
  geometric intersection graphs}, in 12th International Symposium on
  Parameterized and Exact Computation, {IPEC} 2017, September 6-8, 2017,
  Vienna, Austria, 2017, pp.~5:1--5:13.

\bibitem{Bertele72nons}
{\sc U.~Bertele and F.~Brioschi}, {\em Nonserial Dynamic Programming}, Academic
  Press, Inc., Orlando, FL, USA, 1972.

\bibitem{Bodlaender98}
{\sc H.~L. Bodlaender}, {\em A partial {$k$}-arboretum of graphs with bounded
  treewidth}, Theoretical Computers Science, 209 (1998), pp.~1--45.

\bibitem{BodlaenderFLPST16meta}
{\sc H.~L. Bodlaender, F.~V. Fomin, D.~Lokshtanov, E.~Penninkx, S.~Saurabh, and
  D.~M. Thilikos}, {\em ({M}eta) {K}ernelization}, J. {ACM}, 63 (2016),
  pp.~44:1--44:69.

\bibitem{BodlaenderP08}
{\sc H.~L. Bodlaender and E.~Penninkx}, {\em A linear kernel for planar
  feedback vertex set}, in Proceedings of the 3rd International Workshop on
  Parameterized and Exact Computation (IWPEC), vol.~5018 of Lecture Notes in
  Comput. Sci., Springer, 2008, pp.~160--171.

\bibitem{BodlaenderPT08}
{\sc H.~L. Bodlaender, E.~Penninkx, and R.~B. Tan}, {\em A linear kernel for
  the $k$-disjoint cycle problem on planar graphs}, in Proceedings of the 19th
  International Symposium on Algorithms and Computation (ISAAC), vol.~5369 of
  Lecture Notes in Comput. Sci., Springer, 2008, pp.~306--317.

\bibitem{BodlaendervA01a}
{\sc H.~L. Bodlaender and B.~van Antwerpen-de Fluiter}, {\em Reduction
  algorithms for graphs of small treewidth}, Information and Computation, 167
  (2001), pp.~86--119.

\bibitem{BoriePT92}
{\sc R.~B. Borie, R.~G. Parker, and C.~A. Tovey}, {\em Automatic generation of
  linear-time algorithms from predicate calculus descriptions of problems on
  recursively constructed graph families}, Algorithmica, 7 (1992),
  pp.~555--581.

\bibitem{ChekuriC13}
{\sc C.~Chekuri and J.~Chuzhoy}, {\em Polynomial bounds for the {G}rid-{M}inor
  {T}heorem}, in Proceedings of the 46th Annual ACM Symposium on Theory of
  Computing (STOC), ACM, 2014, pp.~60--69.

\bibitem{ChenFKX07}
{\sc J.~Chen, H.~Fernau, I.~A. Kanj, and G.~Xia}, {\em Parametric duality and
  kernelization: Lower bounds and upper bounds on kernel size}, SIAM J.
  Computing, 37 (2007), pp.~1077--1106.

\bibitem{ChenKJ01}
{\sc J.~Chen, I.~A. Kanj, and W.~Jia}, {\em Vertex cover: further observations
  and further improvements}, J. Algorithms, 41 (2001), pp.~280--301.

\bibitem{Chuzhoy15excl}
{\sc J.~Chuzhoy}, {\em Excluded grid theorem: Improved and simplified}, in
  Proceedings of the Forty-Seventh Annual {ACM} on Symposium on Theory of
  Computing, {STOC} 2015, Portland, OR, USA, June 14-17, 2015, R.~A. Servedio
  and R.~Rubinfeld, eds., {ACM}, 2015, pp.~645--654.

\bibitem{Chuzhoy16}
\leavevmode\vrule height 2pt depth -1.6pt width 23pt, {\em Improved bounds for
  the excluded grid theorem}, CoRR, abs/1602.02629 (2016).

\bibitem{Courcelle90}
{\sc B.~Courcelle}, {\em The monadic second-order logic of graphs {I}:
  {R}ecognizable sets of finite graphs}, Information and Computation, 85
  (1990), pp.~12--75.

\bibitem{Courcelle97}
\leavevmode\vrule height 2pt depth -1.6pt width 23pt, {\em The expression of
  graph properties and graph transformations in monadic second-order logic},
  Handbook of Graph Grammars,  (1997), pp.~313--400.

\bibitem{Cygan15_book}
{\sc M.~Cygan, F.~V. Fomin, L.~Kowalik, D.~Lokshtanov, D.~Marx, M.~Pilipczuk,
  M.~Pilipczuk, and S.~Saurabh}, {\em Parameterized Algorithms}, Springer,
  2015.

\bibitem{Fluiter97}
{\sc B.~de~Fluiter}, {\em Algorithms for Graphs of Small Treewidth}, PhD
  thesis, Utrecht University, 1997.

\bibitem{DemaineFHT05sidma}
{\sc E.~D. Demaine, F.~V. Fomin, M.~Hajiaghayi, and D.~M. Thilikos}, {\em
  Bidimensional parameters and local treewidth}, SIAM J. Discrete Math., 18
  (2004), pp.~501--511.

\bibitem{DemaineFHT05talg}
\leavevmode\vrule height 2pt depth -1.6pt width 23pt, {\em Fixed-parameter
  algorithms for $(k, r)$-center in planar graphs and map graphs}, ACM
  Transactions on Algorithms, 1 (2005), pp.~33--47.

\bibitem{DemaineFHT05jacm}
\leavevmode\vrule height 2pt depth -1.6pt width 23pt, {\em Subexponential
  parameterized algorithms on graphs of bounded genus and {$H$}-minor-free
  graphs}, J. ACM, 52 (2005), pp.~866--893.

\bibitem{DemaineHaj05}
{\sc E.~D. Demaine and M.~Hajiaghayi}, {\em Bidimensionality: new connections
  between {FPT} algorithms and {PTAS}s}, in Proceedings of the 16th Annual
  ACM-SIAM Symposium on Discrete Algorithms (SODA), SIAM, 2005, pp.~590--601.

\bibitem{DemaineH05II}
\leavevmode\vrule height 2pt depth -1.6pt width 23pt, {\em Graphs excluding a
  fixed minor have grids as large as treewidth, with combinatorial and
  algorithmic applications through bidimensionality}, in Proceedings of the
  16th Annual ACM-SIAM Symposium on Discrete Algorithms (SODA), ACM-SIAM, 2005,
  pp.~682--689.

\bibitem{DemaineH07-CJ}
\leavevmode\vrule height 2pt depth -1.6pt width 23pt, {\em The bidimensionality
  theory and its algorithmic applications}, Comput. J., 51 (2008),
  pp.~292--302.

\bibitem{DemaineH08}
\leavevmode\vrule height 2pt depth -1.6pt width 23pt, {\em Linearity of grid
  minors in treewidth with applications through bidimensionality},
  Combinatorica, 28 (2008), pp.~19--36.

\bibitem{DornFT08-csr}
{\sc F.~Dorn, F.~V. Fomin, and D.~M. Thilikos}, {\em Subexponential
  parameterized algorithms}, Computer Science Review, 2 (2008), pp.~29--39.

\bibitem{DowneyF99}
{\sc R.~G. Downey and M.~R. Fellows}, {\em Parameterized complexity},
  Springer-Verlag, New York, 1999.

\bibitem{DowneyFbook13}
\leavevmode\vrule height 2pt depth -1.6pt width 23pt, {\em Fundamentals of
  Parameterized Complexity}, Texts in Computer Science, Springer, 2013.

\bibitem{Drange16}
{\sc P.~G. Drange, M.~Dregi, F.~V. Fomin, S.~Kreutzer, D.~Lokshtanov,
  M.~Pilipczuk, M.~Pilipczuk, F.~Reidl, F.~S. Villaamil, S.~Saurabh,
  S.~Siebertz, and S.~Sikdar}, {\em {Kernelization and Sparseness: the Case of
  Dominating Set}}, in Proceedings of the 33rd Symposium on Theoretical Aspects
  of Computer Science (STACS), vol.~47 of Leibniz International Proceedings in
  Informatics (LIPIcs), Dagstuhl, Germany, 2016, Schloss
  Dagstuhl--Leibniz-Zentrum fuer Informatik, pp.~31:1--31:14.

\bibitem{EickmeyerGKKPRS17neig}
{\sc K.~Eickmeyer, A.~C. Giannopoulou, S.~Kreutzer, O.~Kwon, M.~Pilipczuk,
  R.~Rabinovich, and S.~Siebertz}, {\em Neighborhood complexity and
  kernelization for nowhere dense classes of graphs}, in 44th International
  Colloquium on Automata, Languages, and Programming, {ICALP} 2017, vol.~80 of
  LIPIcs, Schloss Dagstuhl - Leibniz-Zentrum fuer Informatik, 2017,
  pp.~63:1--63:14.

\bibitem{FlumGrohebook}
{\sc J.~Flum and M.~Grohe}, {\em Parameterized Complexity Theory}, Texts in
  Theoretical Computer Science. An EATCS Series, Springer-Verlag, Berlin, 2006.

\bibitem{FominGT11cont}
{\sc F.~V. Fomin, P.~A. Golovach, and D.~M. Thilikos}, {\em Contraction
  obstructions for treewidth}, J. Comb. Theory, Ser. B, 101 (2011),
  pp.~302--314.

\bibitem{fomin2015solving}
{\sc F.~V. Fomin, D.~Lokshtanov, N.~Misra, M.~Ramanujan, and S.~Saurabh}, {\em
  Solving $d$-{SAT} via backdoors to small treewidth}, in Proceedings of the
  26th Annual ACM-SIAM Symposium on Discrete Algorithms (SODA), SIAM, 2015,
  pp.~630--641.

\bibitem{FominLMS12}
{\sc F.~V. Fomin, D.~Lokshtanov, N.~Misra, and S.~Saurabh}, {\em Planar
  {F}-deletion: Approximation, kernelization and optimal {FPT} algorithms}, in
  Proceedings of the 53rd Annual Symposium on Foundations of Computer Science
  (FOCS), IEEE, 2012, pp.~470--479.

\bibitem{FominLRS10}
{\sc F.~V. Fomin, D.~Lokshtanov, V.~Raman, and S.~Saurabh}, {\em
  Bidimensionality and {EPTAS}}, in Proceedings of the 21st Annual ACM-SIAM
  Symposium on Discrete Algorithms (SODA), SIAM, 2011, pp.~748--759.

\bibitem{FominLS18excl}
{\sc F.~V. Fomin, D.~Lokshtanov, and S.~Saurabh}, {\em Excluded grid minors and
  efficient polynomial-time approximation schemes}, J. {ACM}, 65 (2018),
  pp.~10:1--10:44.

\bibitem{F.V.Fomin:2010oq}
{\sc F.~V. Fomin, D.~Lokshtanov, S.~Saurabh, and D.~M. Thilikos}, {\em
  Bidimensionality and kernels}, in Proceedings of the 20th Annual ACM-SIAM
  Symposium on Discrete Algorithms (SODA), SIAM, 2010, pp.~503--510.

\bibitem{FominLST18kern}
\leavevmode\vrule height 2pt depth -1.6pt width 23pt, {\em Kernels for
  (connected) dominating set on graphs with excluded topological minors}, {ACM}
  Trans. Algorithms, 14 (2018), pp.~6:1--6:31.

\bibitem{kernelizationbook19}
{\sc F.~V. Fomin, D.~Lokshtanov, S.~Saurabh, and M.~Zehavi}, {\em
  Kernelization. Theory of Parameterized Preprocessing}, Cambridge University
  Press, 2018.

\bibitem{GarneroPST15}
{\sc V.~Garnero, C.~Paul, I.~Sau, and D.~M. Thilikos}, {\em Explicit linear
  kernels via dynamic programming}, {SIAM} J. Discrete Math., 29 (2015),
  pp.~1864--1894.

\bibitem{Garnero2018}
{\sc V.~Garnero, C.~Paul, I.~Sau, and D.~M. Thilikos}, {\em Explicit linear
  kernels for packing problems}, Algorithmica,  (2018).

\bibitem{Gavril74thei}
{\sc F.~Gavril}, {\em The intersection graphs of subtrees in trees are exactly
  the chordal graphs}, J. Comb. Theory, Ser. B, 16 (1974), pp.~47 -- 56.

\bibitem{GiannopoulouPRT17line}
{\sc A.~C. Giannopoulou, M.~Pilipczuk, J.~Raymond, D.~M. Thilikos, and
  M.~Wrochna}, {\em Linear kernels for edge deletion problems to
  immersion-closed graph classes}, in 44th International Colloquium on
  Automata, Languages, and Programming, {ICALP} 2017, July 10-14, 2017, Warsaw,
  Poland, 2017, pp.~57:1--57:15.

\bibitem{GrigorievKT14bidi}
{\sc A.~Grigoriev, A.~Koutsonas, and D.~M. Thilikos}, {\em Bidimensionality of
  geometric intersection graphs}, in SOFSEM 2014: Theory and Practice of
  Computer Science: 40th International Conference on Current Trends in Theory
  and Practice of Computer Science, Nov{\'y} Smokovec, Slovakia, January 26-29,
  2014, Proceedings, V.~Geffert, B.~Preneel, B.~Rovan, J.~{\v{S}}tuller, and
  A.~M. Tjoa, eds., Springer International Publishing, 2014, pp.~293--305.

\bibitem{GuoN07}
{\sc J.~Guo and R.~Niedermeier}, {\em Invitation to data reduction and problem
  kernelization}, SIGACT News, 38 (2007), pp.~31--45.

\bibitem{GuoNW06}
{\sc J.~Guo, R.~Niedermeier, and S.~Wernicke}, {\em Fixed-parameter
  tractability results for full-degree spanning tree and its dual}, in
  Proceedings of the 2nd International Workshop on Parameterized and Exact
  Computation (IWPEC), vol.~4169 of Lecture Notes in Comput. Sci., Springer,
  2006, pp.~203--214.

\bibitem{Halin76sfun}
{\sc R.~Halin}, {\em S-functions for graphs}, Journal of Geometry, 8 (1976),
  pp.~171--186.

\bibitem{KanjPXS08}
{\sc I.~A. Kanj, M.~J. Pelsmajer, G.~Xia, and M.~Schaefer}, {\em On the induced
  matching problem}, in Proceedings of the 25th International Symposium on
  Theoretical Aspects of Computer Science (STACS), vol.~08001, Internationales
  Begegnungs- und Forschungszentrum fuer Informatik (IBFI), Schloss Dagstuhl,
  Germany, Berlin, 2008, pp.~397--408.

\bibitem{KimLPRRSS16}
{\sc E.~J. Kim, A.~Langer, C.~Paul, F.~Reidl, P.~Rossmanith, I.~Sau, and
  S.~Sikdar}, {\em Linear kernels and single-exponential algorithms via
  protrusion decompositions}, {ACM} Trans. Algorithms, 12 (2016),
  pp.~21:1--21:41.

\bibitem{KimST18data}
{\sc E.~J. Kim, M.~Serna, and D.~M. Thilikos}, {\em Data-compression for
  parametrized counting problems on sparse graphs}, in Proceedings of the 29th
  Annual International Symposium on Algorithms and Computation (ISAAC 2003),
  LIPIcs, Schloss Dagstuhl - Leibniz-Zentrum fuer Informatik, Berlin, 2018.

\bibitem{kratsch2014EATCS}
{\sc S.~Kratsch}, {\em Recent developments in kernelization: A survey},
  Bulletin of EATCS, 2 (2014).

\bibitem{LiptonT79}
{\sc R.~J. Lipton and R.~E. Tarjan}, {\em A separator theorem for planar
  graphs}, SIAM J. Appl. Math., 36 (1979), pp.~177--189.

\bibitem{LiptonT80}
\leavevmode\vrule height 2pt depth -1.6pt width 23pt, {\em Applications of a
  planar separator theorem}, SIAM J. Computing, 9 (1980), pp.~615--627.

\bibitem{LokshtanovMS09}
{\sc D.~Lokshtanov, M.~Mnich, and Saurabh}, {\em Linear kernel for planar
  connected dominating set}, in Proceedings of Theory and Applications of
  Models of Computation, (TAMC 2009), Lecture Notes in Comput. Sci., Springer,
  2009.

\bibitem{MisraRS11}
{\sc N.~Misra, V.~Raman, and S.~Saurabh}, {\em Lower bounds on kernelization},
  Discrete Optim., 8 (2011), pp.~110--128.

\bibitem{MoserS07}
{\sc H.~Moser and S.~Sikdar}, {\em The parameterized complexity of the induced
  matching problem in planar graphs}, in Proceedings First Annual International
  WorkshopFrontiers in Algorithmics (FAW), vol.~4613 of Lecture Notes in
  Comput. Sci., Springer, 2007, pp.~325--336.

\bibitem{Niedermeierbook06}
{\sc R.~Niedermeier}, {\em Invitation to fixed-parameter algorithms}, vol.~31
  of Oxford Lecture Series in Mathematics and its Applications, Oxford
  University Press, Oxford, 2006.

\bibitem{RobertsonS3}
{\sc N.~Robertson and P.~D. Seymour}, {\em Graph minors. {III}. {P}lanar
  tree-width}, J. Combinatorial Theory Ser. B, 36 (1984), pp.~49--64.

\bibitem{RobertsonS-V}
{\sc N.~Robertson and P.~D. Seymour}, {\em Graph minors. {V}. {E}xcluding a
  planar graph}, J. Combinatorial Theory Ser. B, 41 (1986), pp.~92--114.

\bibitem{Thilikos12grap}
{\sc D.~M. Thilikos}, {\em Graph minors and parameterized algorithm design}, in
  The Multivariate Algorithmic Revolution and Beyond - Essays Dedicated to
  Michael R. Fellows on the Occasion of His 60th Birthday, 2012, pp.~228--256.

\bibitem{Thilikos15bidi}
\leavevmode\vrule height 2pt depth -1.6pt width 23pt, {\em Bidimensionality and
  parameterized algorithms (invited talk)}, in 10th International Symposium on
  Parameterized and Exact Computation, {IPEC} 2015, September 16-18, 2015,
  Patras, Greece, T.~Husfeldt and I.~A. Kanj, eds., vol.~43 of LIPIcs, Schloss
  Dagstuhl - Leibniz-Zentrum fuer Informatik, 2015, pp.~1--16.

\bibitem{fvs-kernel:talg}
{\sc S.~Thomass{\'e}}, {\em A $4k^2$ kernel for feedback vertex set}, ACM
  Transactions on Algorithms, 6 (2010).

\end{thebibliography}

\end{document}